\documentclass[12pt]{article}
\usepackage[top=1in,bottom=1in,right=1in,left=1in]{geometry}

\usepackage[utf8]{inputenc}
\usepackage{float}
\usepackage{graphicx}  
\usepackage{cite}
\usepackage{natbib}
\graphicspath{ {images/} } 
\usepackage{amsmath}
\usepackage{color}
\usepackage{amsmath}
\usepackage{amsthm}
\usepackage{amsmath}
\usepackage{mathrsfs}
\usepackage{multirow}
\usepackage{amsfonts}
\usepackage{amssymb}
\usepackage{color}
\usepackage{graphicx}  

\usepackage{subfig}
\usepackage[colorlinks]{hyperref}
\hypersetup{
    colorlinks=true,
    linkcolor=red,
    urlcolor=blue,
    linktoc=all,
    citecolor=blue
           }

\newtheorem{theorem}{Theorem}[section]

\newtheorem{lemma}[theorem]{Lemma}
\newtheorem{definition}[theorem]{Definition}
\newtheorem{assumption}{Assumption}[section]
\newtheorem{fact}[theorem]{Fact}
\newtheorem{proposition}[theorem]{Proposition}
\newcommand{\pr}{\mathbb Pr}
\newcommand{\FCO}{\text{``Full Coverage Only''}}
\newcommand{\PCO}{\text{``Partial Coverage Only''}}

\newcommand{\D}{\Delta}

\title{Optimal Pricing Schemes in the Presence of Social Learning and Costly Reporting}
\author{Kaiwei Zhang\footnote{Department of Economics, Penn State University; kqz5253@psu.edu}, Xienan Cheng\footnote{Guanghua School of Management, Peking University; xienancheng@gsm.pku.edu.cn}, Xi Weng\footnote{Guanghua School of Management, Peking University; wengxi125@gsm.pku.edu.cn.} \thanks{*We would like to thank Nageeb Ali, Kalyan Chatterjee, Matthew Elliott, Nima Haghpanah, Yuhta Ishii, Vijay Krishna, Xiao Lin, Jidong Zhou, and various seminar and conference participants for helpful comments and suggestions.} }

\begin{document}
\maketitle

\begin{abstract}
A monopoly platform sells either a risky product (with unknown utility) or a safe product (with known utility) to agents who sequentially arrive and learn the utility of the risky product by the reporting of previous agents. It is costly for agents to report utility; hence the platform has to design both the prices and the reporting bonus to motivate the agents to explore and generate new information. By allowing sellers to set bonuses, we are essentially enabling them to dynamically control the supply of learning signals without significantly affecting the demand for the product. We characterize the optimal bonus and pricing schemes offered by the profit-maximizing platform. It turns out that the optimal scheme falls into one of four types: Full Coverage, Partial Coverage, Immediate Revelation, and Non-Bonus. In a model of exponential bandit, we find that there is a dynamical switch of the types along the learning trajectory. Although learning stops efficiently, information is revealed too slowly compared with the planner's optimal solution.

\medskip


\end{abstract}

\newpage
~\\

\section{Introduction}
One of the most common features of online peer-to-peer market places is that the platform frequently does little up-front screening or certification and
instead tries to maintain quality by using reputation and feedback mechanisms
\citep{Einav2016}. In particular, for the organization of online experience product markets such as houses on Airbnb or drivers on Uber, it crucially requires more consumers' reviews to relieve information asymmetry. Though the internet has made it feasible and common for everyone to leave some feedback after purchasing and experiencing, there is insufficient review supply in the real economy. According to some social surveys about online reviews, nearly 53 $\%$ of all internet users worldwide don't provide any reviews about a product, company, or service within a month\footnote{ https://www.oberlo.com/blog/online-review-statistics}; and only 5$\%$-10$\%$ of consumers are used to writing reviews online\footnote{ https://websitebuilder.org/blog/online-review-statistics/}. 

Meanwhile, it is a common practice on online markets to utilize dynamic pricing and review bonus to motivate the consumers to purchase the experience product and leave feedback. For example, dynamic pricing is common in experience goods markets, such as video games and lead to an increase in both revenue and margins in the pilot categories\footnote{https://www.chartboost.com/blog/how-to-make-dynamic-pricing-work-for-free-to-play-mobile-games/; https://www.mckinsey.com/business-functions/marketing-and-sales/how-we-help-clients/dynamic-pricing}. Many online platforms such as Taobao also offer bonus for consumers' reviews. In the literature, however, dynamic pricing (see, e.g., 
\cite{bergemann1996learning},\cite{bergemann2000experimentation}, \cite{bergemann2006dynamic},\cite{weng2015dynamic}, and \cite{hagiu2020data}) and review bonus (see, e.g., \cite{li2010reputation}  and \cite{tadelis2016buying}) are analyzed separately. It remains unexplored how a profit-maximizing monopolist determine the dynamic price and review bonus together; and how the optimal bonus scheme and dynamic pricing interact with each other. 

In this paper, we study a mechanism design problem faced by a long-lived monopoly platform who sells either a risky product (with unknown utility) or a safe product (with known utility) to a sequence of short-lived risk-neutral consumers. Consumers receive a private signal about the utility of the risky product after purchasing, and incur a reporting cost to give his feedback interpreting as publishing the signal. So a consumer will give a feedback only when he receives a review bonus exceeding his cost, which is also the consumer's private information and known to the consumer before making the purchasing decision. Once a feedback is given, both the platform and the successors will see his private signal; hence there is incomplete social learning in the sense that the common belief about the utility of the risky product is updated according to the published private signals.

In this framework, we investigate how the monopolist commits to different prices of the two products and review bonuses over time to maximize its expected profits. It turns out that the bonus and pricing decisions are not independent and strategically interact with each other, due to following two reasons. First of all, when making the pricing decisions, the monopolist aims to guide the consumers to make the appropriate purchasing decisions. Due to the incompleteness of social learning, the bonus scheme has to be adjusted accordingly. If the monopolist wants to facilitate the learning about the risky product, the bonus should be set very high to guarantee that a lot of private signals are published and vice versa. Second, the bonus scheme also affects the pricing scheme because the bonus received from giving a feedback is also part of the payoff from purchasing the risky product. Therefore, if the review bonus is set very high, the monopolist does not need to lower the price of the risky product too much to motivate the purchase of this product.

We first consider a perfect learning model where the unknown utility of the risky product is perfectly observed after the purchase of this product. In this case, the monopolist faces a trade-off between the value and the cost of information: once knowing the utility of the risky product, he can better guide the consumers to make the appropriate purchasing decisions, but the probability of successfully getting information depends on the bonus he sets. Intuitively, there are two obviously feasible strategies: the monopolist can either set the highest bonus to get the information as soon as possible (Immediate Revelation) if the information is of enough value; or never motivate feedback at all (Non-Bonus) if he thinks the information is of no value. Besides these two extreme strategies, we find that the monopolist may also use a recursive bonus structure, where the optimal bonus decreases over time as the number of remaining consumers decreases and hence the information becomes less valuable. When using this recursive bonus strategy, the monopolist still needs to decide how to guide the purchase the products. Optimally, he may stick to selling the risky product until someone reports its utility (i.e., Full Coverage) or allow the purchase of the safe product as well (i.e., Partial Coverage). When the expected utility from the risky product is sufficiently high but not too high to support Immediate Revelation, the monopolist will prefer Full Coverage over Partial Coverage.


We then extend the basic model to allow imperfect learning via the exponential bandit model (\cite{KRC2005}, \cite{weng2015dynamic}). In a bad state, the risky product generates zero utility for sure, while in a good state, the risky product generates a positive utility with a random arrival time following an exponential distribution.  We characterize the monopolist's optimal pricing and bonus schemes as a function with respect to the belief about the good state. Our previous intuition still holds that the platform uses the optimal schemes to guide consumers to purchase the appropriate products and to reveal information in either moderate or extreme ways. The social learning trajectory in the exponential bandit model is such that the common belief about the good state either remains unchanged if the current agent does not purchase the risky product or make a report, falls because of a revelation of zero utility, or rises directly to 1 after a revelation of positive utility. Consistent with our findings in the basic model, with the decreasing of the common belief, the monopolist will optimally switch from Immediate Revelation to Full Coverage and then to Partial Coverage and Non-Bonus. 
Interestingly, we find that the optimal bonus is monotonically decreasing with respect to the belief, since the marginal value of one more experiment becomes lower when the monopolist believes that the state is more likely to be good. In this situation, although the monopolist expects a very quick revelation of the positive utility, which will immediately make the common belief jump to 1 and lead to even higher profits, he is also satisfied with the current high profits by selling the risky product, and do not want to bear the loss of bonus and the possible decrease of belief. As a result, a higher common belief leads to a lower optimal bonus. 

We conduct welfare analysis in the imperfect learning case, where we assume a social planner who can commit to an allocation rule and reporting rule, based on the common belief and agents' reporting cost. We prove that social planner will act like the monopolist, and the optimal strategy is again one of the possible optimal strategies used by the monopolist. A surprising result is that unlike traditional literature, there is no inefficient earlier termination of social learning, since the monopolist can always take advantage of Partial Coverage by setting a low enough bonus to encourage partial agents to reveal the information cheaply.  
Though stop at the same belief, we find that the social planner may still lead a more efficient experiment than the monopolist, in the sense that the social planner will encourage reporting with higher probability due to lower cost of experiments. The comparative static analysis shows that a higher learning speed  will lead to more experimentation of the risky product, and hence lower switching belief from Full Coverage to Partial Coverage. A higher utility from the safe product or a higher discount rate may lead to inefficiently earlier termination of the learning process. 
\section{Literature Review}

The literature on informational cascades that originated
with the work of \cite{banerjee1992simple} and \cite{bikhchandani1992theory} is
probably the closest to the model presented here. In these herding models, agents may directly follow actions ahead of him without referring to his own signal, which makes public information swamps private information and stops being updated. \cite{bose2006dynamic} and \cite{bose2008monopoly} study the monopoly pricing problems in herding situations. Our problem is different from their work because in our model, agents are not endowed with his own signal. With this assumption, \cite{kremer2014implementing} considers a social planner who can reveal information strategically to agents, in order to persuade agents to optimally do experiments. They assume the planner can take advantage of information asymmetry to exert Bayesian Persuasion, hence the possiblility of exploitation makes up to the experiment of agents' experiment. However, like e-commerce platforms, consumers' reviews are public to both the seller and the buyers, which makes Bayesian Persuasion infeasible. To our knowledge there is no relevant study, hence we try to fill this gap.

Agents in our model must choose from two arms like in a bandit problem setting originated with the work of \cite{rothschild1974bandit}. Unlike the single agent learning problem there, we assume there are agents involved sequentially, along the lines with works of \cite{bolton1999strategic}. In multi-agent setting with public information, the positive externality often rises free-rider problems, and social learning is always terminated insufficiently ( \cite{bergemann1996learning}, 
\cite{bergemann2000experimentation}, \cite{bergemann2006dynamic},\cite{hagiu2020data}). \cite{weng2015dynamic} consider the dynamic pricing with individual learning, where every agent may have his own type to effect his action, though of no value to learn by others. Different to his work, our paper considers the social learning part, where agents care about the common utility of the risky arm.

The literature of review motivation is also highly related to our work. Though it's with no doubt that online reviews help consumers to decrease the uncertainty of purchasing some certain goods, we can see that current review provision is lacking. Some works are about dealing with missing reviews and ratings (see \cite{ying2006leveraging} and \cite{he2016fast}), while some more related works focus on motivating more reviews from consumers (see \cite{li2010reputation}  and \cite{tadelis2016buying}). The idea of motivating reviews may be originated from the work of \cite{avery1999market}, where reporting cost is endogenously generated because of externality of information. When facing reporting cost, 
\cite{krishna2012voluntary} find that full participation may not be optimal and if full participation is compulsory, every one's truthful revelation may be not a Nash Equilibrium. \cite{osborne2020information} find that agents' threshold strategy limits the probability of revelation even when there are a large number of agents. But when transferring is plausible, agents' behavior may make their payoff attach first best. In our work, we use the monetary encouragement to cover agents' reporting cost in order to complete social learning.

Multiple studies have shown that monetary incentives are effective at encouraging wanted behavior online. For example, experimental studies show that financial incentives are effective
in motivating consumers to write reviews on online consuming platforms (\cite{fradkin2015bias}, \cite{cabral2015dollar}, \cite{khern2018extrinsic}). It is important to note that financial incentives offered in most such studies have generally been quite small, like \$ 1 or 2, or 0.5 RMB (\cite{cabral2015dollar}), \$ 0.04 in Amazon Mechanical Turk or 10 RMB in a Chinese online cloth retailing platform  (\cite{burtch2018stimulating}). Empiric works also show that paying for reviews do not effect consumers' purchasing behavior or the demand curve faced by the seller (\cite{cabral2015dollar}), and the feedback motivated by financial incentives are
often upward biased since consumers feel ``bribed" (\cite{cabral2015dollar}).
\section{Model Set Up}
\label{section_Imperfect_Learning}
There is a monopolistic seller offering two distinct products, denoted as $i=1,2$ (referred to as product 1 and 2 henceforth), to a sequence of risk-neutral agents in an infinite horizon, each with a lifetime of $\Delta\to0$ periods. In each period $\Delta$, a single agent is present. The production cost for each product is normalized to zero. The seller, being long-lived, aims to maximize its discounted profit over time, while short-lived agents, having outside utility of 0, seek to maximize their expected utility within period $\D$ by selecting at most one product from the available options. Product 2 is considered a safe product with a known utility, while product 1 is categorized as a risky product, whose utility is unknown to both buyers and the seller. In line with previous literature, such as \cite{KRC2005}, we label product 2 as the "safe arm" and product 1 as the "risky arm". 

Choosing the safe arm results in a deterministic flow utility of $s>0$, while the utility of choosing the risky arm is modeled in the following bandit setting: purchasing the risky arm can yield either a lump-sum value $z$ (good news) or zero utility, contingent on a hidden state of the world unknown to both the seller and buyers. There are two states, denoted as good (G) or bad (B). In the good state, the arrival rate of realizing value $z$ follows an exponential distribution with parameter $\lambda$, i.e., the probability density function is $\lambda e^{-\lambda t}$. In the bad state, arm 1 always results in zero utility. After purchasing the risky arm, buyers can privately learn their own utility under the risky arm: either a lump-sum $z$ or zero utility. Since $z$ can only be achieved under the good state, we assert that this good news is conclusive. We introduce a crucial assumption \ref{assumption_meaningful} here to ensure our discussion is non-trivial; otherwise, the risky arm is dominated by the safe arm.
\begin{assumption}
	$g\equiv \lambda z>s>0$
	\label{assumption_meaningful}
\end{assumption}

After purchasing and privately observing their own utility, buyers can decide whether to disclose their news to the public (they cannot privately report to the seller to create information asymmetry). Due to their short-lived nature, there is no strategic incentive for them to lie or conceal information. However, the act of reporting itself incurs a cost, with each buyer having a private reporting cost $c_i$, independently and identically distributed according to the distribution function $H(\cdot)$ with support $[0,\bar{c}]$. For simplicity, we assume that $H(\cdot)$ is strictly increasing and has a smooth density function $h(\cdot)$. Moreover, we impose the following standard assumptions \ref{Inada} and \ref{assumption:monotone_beta} on the density function $h(\cdot)$ and virtue value  $\beta(x)=x+\frac{H(x)}{h(x)}$:
\begin{assumption}
\label{Inada}
The density function $h(\cdot)$ satisfies:
$\lim\limits_{x\to 0^+} h(x)>0$, and  $\lim\limits_{x\to c^-} h(x)>0$.
\footnote{A more general assumption for part 2 in Assumption \ref{Inada} is that there exists some order $p$ such that $\lim\limits_{x\to \bar c^-}h^{(p)}(x)\neq 0$ if $\lim\limits_{x\to \bar c^-}h^{(k)}(x)= 0,\forall k<p$.Here if $\lim\limits_{x\to \bar c^-}h(x)=0$, form L Hospital's Rule $\lambda(\bar c)=M+\bar c+\lim\limits_{x\to \bar c^-}\frac{H(x)-H^2(x)}{h(x)}=M+\bar c-\lim\limits_{x\to \bar c^-}\frac{h(x)}{h'(x)}$ so if $\lim\limits_{x\to \bar c^-}h'(x)=0$ we can still get the same result by using  $\lim\limits_{x\to \bar c^-}h''(x)>0$ or so on. All we need is that the limit $\lim\limits_{x\to \bar c^-}\frac{h(x)}{h'(x)}$ is convergent to zero hence as long as there exist some order $p$ such that $\lim\limits_{x\to \bar c^-}h^{(p)}(x)\neq 0$ if $\lim\limits_{x\to \bar c^-}h^{(k)}(x)= 0,\forall k<p$, which always exists unless $h(x)$ is constant in a neighborhood of $\bar c$. }
\end{assumption}

\begin{assumption}
$\beta(x)$ is continuously differentiable and 	$\beta'(x)$ is always positive.
 \label{assumption:monotone_beta}
\end{assumption}
Since, in our setting, $H(\cdot)$ has support with an upper bound of $\bar{c}$, a potential strategy is to employ the "Immediate Revelation" strategy, which involves setting a bonus of $\bar{c}$ to ensure that every buyer reports at the outset. We introduce assumption \ref{assumption_barc_large} to rule out this strategy. Further discussion on this matter will be presented in the appendix.
\begin{assumption}
	$\bar c>(\frac{\lambda}{r}+1)(g-s)$
 \label{assumption_barc_large}
\end{assumption}

We assume that each buyer can observe the entire purchasing and reporting history. This assumption allows us to concentrate on the pricing and report motivating problems. Although this assumption is evidently strong, we find that this situation is quite common in online markets, where potential consumers have full access to previous consumers' feedback. If consumers' reporting is privately observed by the seller, information design problems arise. This allows the seller to exploit a dynamic Bayesian persuasion strategy to motivate agents to purchase the risky arm (see \cite{kremer2014implementing}). 
In this scenario, society will collectively form a common belief regarding the probability that the state is good. Denote common purchasing and reporting history from past observations before time $t$ as
$\mathscr{F}(t)$, a natural state variable is the common belief $\alpha_t=\Pr(\theta=G|\mathscr{F}(t))$, which is the common posterior probability that $\theta=G$.

\textbf{Timing.} We denote the agent residing in the interval $[t,t+\Delta)$ as agent $t$. $p_{1t}$ and $p_{2t}$ represent the prices for the risky arm and the safe arm at time $t$ respectively. If purchasing happens, the prices $p_{it}$ are paid at the beginning of this period. Moreover, if he decides to report when the bonus is $\tilde b_t$, he pays a deterministic flow cost $c_i$, and get bonus $\tilde b_t$ at the end of period $\Delta$\footnote{Hence the consumer has no incentive to report saying a good news earlier since the bonus is paid at the end of $\Delta$ no matter what he reports.}.
Thus, at time $t$, agent $t$ observes $p_{1t}$, $p_{2t}$, $\tilde{b}_t$, and the common belief $\alpha_t$ formed by the entire purchasing and reporting history before $t$.  Additionally, we consider a common exponential discount factor with discount factor $r$; that is, a payoff $v$ received at time $t$ generates a utility of $e^{-rt} v$ at time 0 both for the buyer and seller.

\textbf{Buyer's Surplus.} If agent $t$ purchases safe arm, he would have surplus \[\int_0^{\Delta} e^{-rt}s  dt-p_{2t} \]
If he chooses the risky arm, he may get the lump sum utility $z$ which may realize at any time $t$ with density $\lambda e^{-\lambda t}$ within $\Delta$. That is to say, if the consumer chooses risky arm, his expected surplus is
\[\int_0^{\Delta} e^{-rt}\lambda e^{-\lambda t} \alpha z dt-p_{1t} \] if he doesn't report. Otherwise, if he reports, the expected surplus is 
\[\int_0^{\Delta} e^{-rt}\lambda e^{-\lambda t} \alpha z  dt-p_{1t}-\int_{0}^{\Delta} e^{-rt} c_i dt+e^{-r\Delta} \tilde b_t\]

\textbf{Seller's Policy.}
The seller's problem is to determine the pricing and bonus scheme. We permit the seller to commit to a contingent bonus scheme $\tilde{b}_t(\alpha_t)$ and pricing scheme $p_{1t}(\alpha_t), p_{2t}(\alpha_t)$ before the game starts. However, we do not allow bonuses to be contingent on the polarity of the feedback. This assumption rules out the trivial strategy of bribing for reputation.  It's also a reasonable assumption, since in most online purchasing platforms, the seller has no access to forbid consumers to provide negative feedback or complaints.  Without loss of generality, we focus on pure strategies. Moreover, we limit our attention in Markovian strategies and equilibrium.

Within period $[t,t+\D)$, the seller can only sell at most one of the two products to agent $t$. Since he is the monopolist and buyer's outside utility is 0, the seller can always extract all the consumer surplus and sell whatever product he wants, by decreasing the target price with $\epsilon\to 0^+$ to create a positive surplus. For example, the seller can set $p_2(\alpha_t)=\frac{1-e^{-r\Delta}}{r}s$ and $p_1(\alpha_t)=\frac{\lambda(1-e^{-(r+\lambda)\Delta})}{\lambda+r}\alpha_t z-\epsilon$ to guide all buyers to buy risky arm. The seller can also $p_1=\frac{\lambda(1-e^{-(r+\lambda)\Delta})}{\lambda+r}\alpha_t z+\epsilon$ to guide all buyers who won't report to buy safe arm. 

Without loss of generality\footnote{For example, in \cite{bergemann2000experimentation}, safe product and risky product are provided by two sellers, and they consider the Bertrand competition to do tie breaking.}, we assume the seller will price fairly and consumers do tie breaking according to seller's policy $\ell(\alpha_t)$, where seller's policy $\ell:[0,1]\to \{1,2\}$ stands for which product he would sell to agents who doesn't report. 
In other words, we always set
\begin{equation}
\begin{aligned}
p_{1t}=&\int_0^{\Delta} e^{-rt}\lambda e^{-\lambda t} \alpha z dt=\frac{\lambda(1-e^{-(r+\lambda)\Delta})}{\lambda+r}\alpha z\\
        p_{2t}=&\int_0^{\Delta} e^{-rt}s  dt=\frac{1-e^{-r\Delta}}{r}s
    \label{price}
\end{aligned}
    \tag{Fair Pricing}
\end{equation}
and allows the seller to implement a common belief contingent bonus scheme $b_t(\alpha_t)$ and selling policy $\ell(\alpha_t)$. If the seller sets $\ell(\alpha_t)=2$ in the interval $[t_1,t_2]$, we can expect that, along the equilibrium path, almost all agents will opt for the safe arm, with only a fraction choosing the risky arm and submitting a report. Therefore, we label the case where $\ell(\alpha_t)=1$ as ``promoting the risky arm" and $\ell(\alpha_t)=2$ as ``promoting the safe arm" for simplicity. 

The seller should also decide the bonus scheme, which is a measurable function $\tilde b:[0,1]\to \mathbb R$, where $\tilde b(\alpha_t)$
stands for the bonus he would pay to agent $t$ at time $t+\D$, if agent $t$ reports his utility of purchasing the risky arm.  
The seller's stage payoff within $[t,t+\D)$ depends on his action and agent $t$'s action:
\[  \pi(\D,b_t,\ell_t)=\left\{
\begin{aligned}
	& p_{1t} & \text{ if } \ell_t=1,r_t=0 \\
	& p_{1t} - e^{-r\D} \tilde b_t & \text{ if } \ell_t=1,r_t=1\\
	& p_{2t} &\text{ if } \ell_t=2
\end{aligned} \right.
\label{equation_seller_payoffl} \]
~\\
~\\
After simplification, the dynamic game consists of three components: the seller's promoting  policy $\ell:[0,1]\to \{1,2\}$ and bonus scheme $\tilde b:[0,1]\to \mathbb R$ , and the buyers' reporting strategies $r_t:[0,\bar c]\times \mathbb R\to \{0,1\}$. We would define our equilibrium in \ref{def_mpe} after we define the seller's value function.

\section{Model Analysis}
\label{section_Model_Analysis}
Our main result is Proposition \ref{optimal_switching_proposition} and \ref{proposition_speed}. Proposition \ref{optimal_switching_proposition} characterize the monopolist's optimal linkage between pricing and bonus setting, as well as the optimal switching between strategies and optimal stopping rules. Proposition \ref{proposition_speed} shows that the learning process led by the monopolist will terminate at the first best cutoff, while the learning speed is inefficiently low.  
\subsection{Seller's Strategy}
\label{Section_preliminary_analysis}
In general, there are four possible combinations of this selling policy and bonus scheme, as shown in Table \ref{table_strategy}.

\begin{table}[h]
    \centering
    \begin{tabular}{|c|c|c|}\hline
         within $[t,t+\D)$: &promoting risky arm  &promoting safe arm \\\hline
         $\bar c>b_t>0$& Full Coverage & Partial Coverage\\ \hline
         $b_t=0$& Non Bonus & Safe Arm\\ \hline
    \end{tabular}
    \caption{$2\times 2$ strategies}
    \label{table_strategy}
\end{table}

By setting $b_t=0$, the promoting policy uniquely determines the outcome. If the seller is promoting risky/safe arm, for sure the current agent $t$ will purchase the risky/safe arm. And there will be no learning because of no reporting, so we denote these stage strategies as ``Non Bonus'' (NB) and ``Safe Arm'' (SA) separately. The stage payoff of using NB and SA are:
\[\pi_{\text{NB}}(\D,b_t=0)=p_{1t}, \pi_{\text{SA}}(\D,b_t=0)=p_{2t}\]

By setting some $b_t>0$ and mainly selling risky arm, the seller has some probability to motivate the current agent to report, otherwise the current agent purchases the risky arm and not report. In this case, the seller uses risky arm to fully cover the market demand, hence we denote this strategy as ``Full Coverage'' (FC). Under Full Coverage, the agent with reporting cost $c_i$ will purchase the risky arm for sure and report if and only if 
\begin{equation}
    e^{-r\Delta}\tilde b_t-\int_0^{\Delta} e^{-rt } c_i dt\ge 0\label{report_condition}
\end{equation}
Hence by setting $\tilde b_t$ (or equivalently, setting $b_t\equiv \frac{r}{1-e^{-r\Delta}} \tilde b_t$), the ex-ante probability of successfully motivating the reporting is 
\[\pr(\text{report})=H(\frac{e^{-r\Delta}r}{1-e^{-r\Delta}} \tilde b_t)\equiv H(e^{-r\Delta} b_t)\]
Hence for the seller the immediate surplus is then 
\begin{equation}
\begin{aligned}
        \pi_{\text{FC}}(\Delta,b_t) &=p_{1t}-e^{-r\Delta}\tilde b_t\pr(\text{report}) \\&= \frac{\lambda}{r+\lambda} (1-e^{-(r+\lambda)\Delta})\alpha z -H(e^{-r\Delta}b_t) b_t\frac{1-e^{-r\Delta}}{r} e^{-r\Delta}
\end{aligned}
    \label{pi_FC_Delta}
\end{equation}

If the seller is mainly selling safe arm, then positive bonus may affect market demand. If the current agent has reporting cost $c_i$ and 
\begin{equation}
    \int_0^{\Delta} e^{-rt}\lambda e^{-\lambda t} \alpha z  dt-p_{1t}-\int_{0}^{\Delta} e^{-rt} c_i dt+e^{-r\Delta} \tilde b_t\ge \int_0^{\Delta} e^{-rt}s  dt-p_{2t} 
    \label{report_condition_PC}
\end{equation}
the current agent will purchase the risky arm and report  for sure. Otherwise he will purchase the safe arm in favor of partial coverage policy. Hence on the equilibrium path, all agents who doesn't report would purchase the safe arm, while those who purchases the risky arm would report for sure.  We call this strategy ``Partial Coverage'' (PC) in the sense that the seller is still learning the quality of risky arm, but only using it partially covering the market demand. 
The simplified expression of \eqref{report_condition_PC},  incorporating pricing and tie-breaking, leads to the same condition as \eqref{report_condition}. This indicates that, given an identical bonus scheme, agents who would report under Full Coverage will do so if and only if they would also report under Partial Coverage. So the probability of motivating the current agent to report would be the same.
Hence we can get
\begin{equation}
    \pi_{\text{PC}}(\Delta,b_t)=\pr(\text{report}) [p_{1t}-e^{-r\Delta}\tilde b_{t}] +[1-\pr(\text{report})] p_{2t}
    \label{pi_PC_Delta}
\end{equation}

\subsection{Seller's Value Function}
If it's already known that the state is good, value function of keeping selling $R_1$ will be:
\[
\tilde V=\int_0^{\Delta} e^{-rt} (\lambda e^{-\lambda t})z dt+e^{-r\Delta}\tilde V
\]
Let $\Delta \to 0$ we can get $\tilde V=\frac{\lambda z}{r}$. Similarly, keeping selling $R_2$ leads to payoff $\tilde U=\frac{s}{r}$, which is the value function of using Safe Arm strategy. With the same logic, the value function of using Non Bonus strategy should be $\frac{\alpha \lambda z}{r}$.

Plug $\D\to0$ in equation \ref{pi_FC_Delta} and \ref{pi_PC_Delta}, we can get
\begin{equation}
    \lim_{\Delta\to 0} \frac{\pi_{\text{FC}}(\Delta,b_t)}{\Delta} = \lambda \alpha z-H(b_t)b_t
    \label{pi_fc}
\end{equation}
and 
\begin{equation}
     \lim\limits_{\Delta\to 0}\frac{\pi_{\text{PC}}(\Delta,b_t)}{\Delta}
    = H(b_t) [\lambda\alpha z-b_t] +[1-H(b_t)] s
    \label{pi_pc}
\end{equation}

As a conclusion of equation \eqref{pi_fc} and \eqref{pi_pc}, let $\pi_{L}(\Delta,b_t)$ stands for seller's expectation of immediate net profit within period $\Delta$ using strategy $L\in \{\text{FC,PC,NB,SA}\}$. If seller's policy at time $t$ is $(L,b_t)$ then the period profit fits
\begin{equation}
    \lim_{\Delta\to0}\frac{\pi_{L} (\Delta,b_t)}{\Delta}=\left\{
\begin{aligned}
	&\lambda \alpha z-H(b_t)b_t & \text{ if } L=\text{FC}\\
	& H(b_t) [\lambda \alpha z-b_t] +[1-H(b_t)] s & \text{ if } L=\text{PC}\\
	& \lambda \alpha z &\text{ if } L=\text{NB}\\
	& s &\text{ if } L=\text{SA}
\end{aligned} \right.
\label{equation_pi_1l}
\end{equation}
And the probability of reporting the utility of purchasing risky arm is
\[P_L(b_t) = \lim_{\Delta\to0} P_L(\Delta, b_t)=\left\{
\begin{aligned}
	&H(b_t)& \text{ if } L=\text{FC or PC}\\
	&0 &\text{ if } L=\text{NB or SA}
\end{aligned} \right.\]

Hence the seller's value function $\Pi(\cdot)$ follows the dynamic process:
\begin{equation}
\begin{aligned}
    \Pi(\alpha)=&\max\limits_{L,b} \pi_{L} (\Delta,b_t)+P_L(\Delta, b)e^{-r\Delta}\gamma(\alpha,\Delta) \Pi(\alpha')+\\
    &P_L(\Delta, b) e^{-r\Delta}(1-\gamma(\alpha,\Delta)) \tilde V
    +(1-P_L(\Delta, b))e^{-r\Delta}\Pi(\alpha)
\end{aligned}  
    \label{equation_max}
\end{equation}
where \[\gamma(\alpha,\Delta)=1-\alpha(1-e^{-\lambda \Delta})\] stands for the probability that in time period $\Delta$ the current agent reports a zero utility and then leads to belief updating to $\alpha'$, and  $1-\gamma(\alpha,\Delta)$ stands for the probability that within $\Delta$ the agent reports a utility $z$ which makes the common belief jump to 1 immediately. 
$P_L(\Delta,b)=\pr(\text{report})$ stands for the probability of motivating the current agent to report successfully under strategy $L$. The common belief will follow the Bayesian updating rule
\[
\alpha'=\frac{\alpha e^{-\lambda \Delta}}{\alpha e^{-\lambda \Delta}+1-\alpha}
\]

Let $\D\to 0$ (for more details, see Appendix \ref{appendix_proof_pre}), we can get $\Pi(\alpha)$ actually follows the ordinary differential equation based on equation (\ref{equation_max}):
\begin{equation}
    r \Pi(\alpha)=\max\limits_{L,b}\lim_{\Delta\to0}\frac{\pi_{L} (\Delta,b_t)}{\Delta}+P_L(b) \frac{\alpha \lambda^2 z }{r}-P_L(b)\alpha \lambda \Pi(\alpha)+P_L(b)(\alpha^2-\alpha)\lambda \Pi'(\alpha)
    \label{ODE_vf}
\end{equation}

As a conclusion, the seller's strategy profile ($\ell,\tilde b$) is equivalent to set $(L,b)$. With this notation, we define the equilibrium as follows:
\begin{definition}{(Markov Perfect Equilibrium)} A collection of strategies $\{\ell^*,b^*,r_t^*\}$ is a Markov Perfect Equilibrium if:
\begin{enumerate}
    \item Agent $t$ maximizes his utility, i.e., \eqref{report_condition} holds, $\forall t$.
    \item Seller maximizes his lifetime utility defined in \eqref{ODE_vf}.
\end{enumerate}
\label{def_mpe}
\end{definition}

In the left of this section, we characterize the optimal pricing and bonus scheme for the monopolist. To do so we first derive the optimal pricing and bonus scheme for the monopolist with bounded rationality, in the sense that he can use Full Coverage or Partial Coverage, with optimally stopping experiments and switching to Non Bonus and Safe Arm. As a main result here, Fact \ref{Fact_intersect} shows that these two value functions will intersect with a kink, hence the fully rational monopolist should optimally switching between Full Coverage and Partial Coverage, which is concluded in Proposition \ref{optimal_switching_proposition}. Then we characterize the optimal pricing and bonus scheme for a social planner for welfare analysis in terms of efficient termination and learning speed.
 
 \subsection{``Naive'' Monopolist}
 We first consider a simple scenario where the seller can only use one of two interior strategies: either $\FCO$ or $\PCO$. Additionally, the seller knows how to optimally switch between these interior strategies to the corner solution, namely, the Non-Bonus or Safe Arm strategy. We refer to this type of seller as a ``naive'' monopolist. In contrast, a fully rational seller who knows how to optimally switch among all strategies is termed an ``optimal switching monopolist''. Introducing the ``naive'' monopolist has its advantages, as analyzing its value function helps us study the possible patterns and directions for optimal switching.

 Firstly we consider ``Full Coverage Only'' monopolist (i.e. the seller can only use Full Coverage, Safe Arm or Non bonus). We denote $V(\alpha)$ as the value function of using Full Coverage. Combine equation (\ref{equation_pi_1l}) and (\ref{ODE_vf}), we can get 
 \begin{equation}
    r V(\alpha)=\max\limits_{b\in[0,\bar c]} \lambda \alpha z[1+H(b)\frac{\lambda}{r}]-H(b)b-H(b)\alpha \lambda V(\alpha)+H(b)(\alpha^2-\alpha)\lambda V'(\alpha)
    \label{result ode}
 \tag{FC-ODE}
\end{equation}
 which is a standard Bellman Equation. Followed by principal of optimality (\cite{SLP}), the value function $V(\cdot)$ is well-behaved, in the sense of continuous, differentiable, and convex. Since we only consider the possibility that the monopoly can stop doing Full Coverage experiments and switch to the Safe Arm or Non Bonus, the optimal stopping problem can be easily captured by smooth pasting condition. The following lemma describes the strategy under this setting:
 
 \begin{lemma}
 The ``Full Coverage Only'' monopolist uses a cutoff strategy:
 there exists $0<\underline \alpha_1<\hat \alpha_{FC}^{NB}<1$ such that he uses NB for $\alpha>\hat \alpha_{FC}^{NB}$, uses 
 SA for $\alpha<\underline \alpha_1$. Otherwise he uses FC and sets $b_t(\alpha)>0$ for $\underline \alpha_1\le \alpha \le \hat \alpha_{FC}^{NB}$. Denote $\underline b_1$ as the bonus at belief $\underline \alpha_1$, then $b_t(\alpha),V(\alpha)$ and $\underline \alpha_1,\underline b_1$ fit the following relationships:
 	\begin{enumerate}
 		\item 	First order condition: \[
 		\beta(b)=\frac{\alpha z \lambda^2 }{r}+(\alpha^2-\alpha)\lambda V'(\alpha)-\alpha\lambda V(\alpha)
 		\]
 		\item Bellman equation: \[
 		rV(\alpha)=\alpha z \lambda+\frac{H^2(b)}{h(b)}
 		\]
 		\item Optimal Stopping: \[\left\{
 		\begin{aligned}
 			& \beta(\underline b_1)=\frac{\underline \alpha_1  \lambda}{r} (z\lambda-s) \\
 			& s =\underline \alpha_1  z \lambda+ \frac{H^2(\underline b_1)}{h(\underline b_1)}
 		\end{aligned} \right.\]
 		\item Monotonically Decreasing: $b(\alpha)$ is monotonically decreasing with respect to $\alpha$. 
 	\end{enumerate}
  \label{fconly_lemma}
 \end{lemma}
 
Lemma \ref{fconly_lemma} claims that the monopolist is less willing to motivate agents' learning when the common belief is higher. Intuitively, there are two powers affecting bonus setting. When the belief is higher, the monopolist would think it's more possible to jump to 1 and hence get higher further profit, which provides high incentive to set high bonus (exploration). However, if the belief is already high enough, the seller takes advantage of high common belief, hence setting high bonus may increase the risk that the common belief drops under reporting, which hurts the future profit (exploitation). Lemma \ref{fconly_lemma} shows the latter power overcomes the first one when $\alpha$ is in the area that using Full Coverage is optimal. Actually as we will see later, the first concern can be solved by using PC. A natural result is
 \[s>\underline{\alpha_1} z \lambda\]
 which means the $\FCO$ monopolist would like to encourage agents to test the risky arm even when it yields lower expected utility than safe arm. That's the natural result of leading a social learning process, even though the monopolist can not optimally switch between all potential strategies. In the formal model, the seller may switch from FC to PC (instead of SA), which will be discussed in section \ref{section_optimal_switching}.
 
 Last but not least, we should emphasis the possibility of using Non Bonus. From interior solution got from part 1 in Lemma \ref{fconly_lemma}, Full Coverage will generate zero bonus when 
 \[\frac{\alpha z \lambda^2}{r} +(\alpha^2-\alpha)\lambda V'(\alpha)-\alpha\lambda V(\alpha)\le 0\]
 Combined with part 2 in Lemma \ref{fconly_lemma},it's equivalent (in the interior solution part) to 
 \begin{equation}
     \alpha> \hat\alpha^{NB}_{FC}\footnote{Here we use $\hat\alpha^{NB}_{FC}$ since it's not the true switching cutoff for the optimal switching monopolist, since the value function for the optimal switching monopolist has the different boundary condition, even though they fit the same relationship between $V'$ and $V$ mentioned in \ref{fconly_lemma}.}
     \label{Full Coverage Non Bonus}
 \end{equation}

 where $ \hat\alpha^{NB}_{FC}$ is the solution of $V(\alpha)=\frac{\alpha z \lambda}{r}$.   Equation  \ref{Full Coverage Non Bonus} shows that the monopolist would prefer Non Bonus to Full Coverage for $\forall \alpha\in(\hat\alpha^{NB}_{FC},1)$. In other words, if the temporary common belief $\alpha_t\in(\hat\alpha^{NB}_{FC},1)$, the monopolist will sell risky arm to each agent and won't motivate them to report the utility. As a result, the common belief would never change since no social signals are generated in this case. As mentioned before, the monopolist would set lower bonus when common belief is higher, given the common belief $\alpha<\hat\alpha^{NB}_{FC}$. In this period, the monopolist is motivating agents to do more experiments). Otherwise when $\alpha\in(\hat\alpha^{NB}_{FC},1)$, the monopolist would not set any belief though it's possible that some agents come and report a positive utility and push the common belief directly to 1, since he can sell the risky arm 
and take advantage of the pretty high common belief.

 Secondly, we consider $\PCO$ monopolist setting. As before, the value function $U(\alpha)$ must fit the following ordinary differential equation:
 \[
 r U(\alpha)=\max_b H(b) (\lambda \alpha z-b) +[1-H(b)]s+H(b)\frac{\alpha \lambda^2 z}{r}+H(b)(\alpha^2-\alpha)\lambda U'(\alpha)-H(b)\alpha \lambda U(\alpha)
 \tag{PC-ODE}
 \]
 
 And the following lemma characterize the solution of optimal stopping problem of switching from $U(\cdot)$ to $\tilde U$

 \begin{lemma}
 The ``Partial Coverage Only'' monopolist uses a cutoff strategy:
 there exists $0<\underline \alpha_2<1$ such that he uses PC and sets $b_t(\alpha)>0$ for $\alpha>\underline \alpha_2$, and uses 
 SA for $\alpha<\underline \alpha_2$. Denote $\underline b_2$ as the bonus at belief $\underline \alpha_2$, then $b_t(\alpha),U(\alpha)$ and $\underline \alpha_2,\underline b_2$ fit the following relationships:
 	\begin{enumerate}
 		\item 	First order condition: 
            \[ \beta(b)=\lambda \alpha z -  s+\frac{\alpha z \lambda^2}{r}+(\alpha^2-\alpha)\lambda U'(\alpha)-\lambda \alpha U(\alpha)
 		\]
 		\item Bellman equation: \[
 		r U(\alpha)=s+\frac{H^2(b)}{h(b)} 
 		\]
 		\item Optimal Stopping: \[\left\{
 		\begin{aligned}
 			&\underline b_2=0\\
 			&\underline \alpha_2=\frac{rs}{\lambda(g-s)+rg}
 		\end{aligned} \right.\]
 		\item Monotonically Increasing: $b(\alpha)$ is monotonically increasing with respect to $\alpha$. 
 	\end{enumerate}
  \label{pconly_lemma}
 \end{lemma}

A noticeable fact is that in Partial Coverage the net bonus $R$ is monotonically increasing with respect to $\alpha$, unlike Full Coverage. As a result, $\PCO$ will not optimally switch to Non Bonus, even though he is allowed in this setting. A more detailed discussion about this fact is illustrated in the Figure \ref{fig:pqplot} later. The other fact is that $\underline b_2=0$. We will see in section \ref{section_welfare} that this result will be the key force driving our main welfare analysis.
Lemma \ref{fconly_lemma} and \ref{pconly_lemma} together describe the value function of the Naive Monopolist. A quick illustration is in Figure \ref{wrong_fc_pc}.  
 \begin{figure}[H]
 	\centering  
 	{\label{Fig.sub.1}
 	\includegraphics[width=0.5\textwidth]{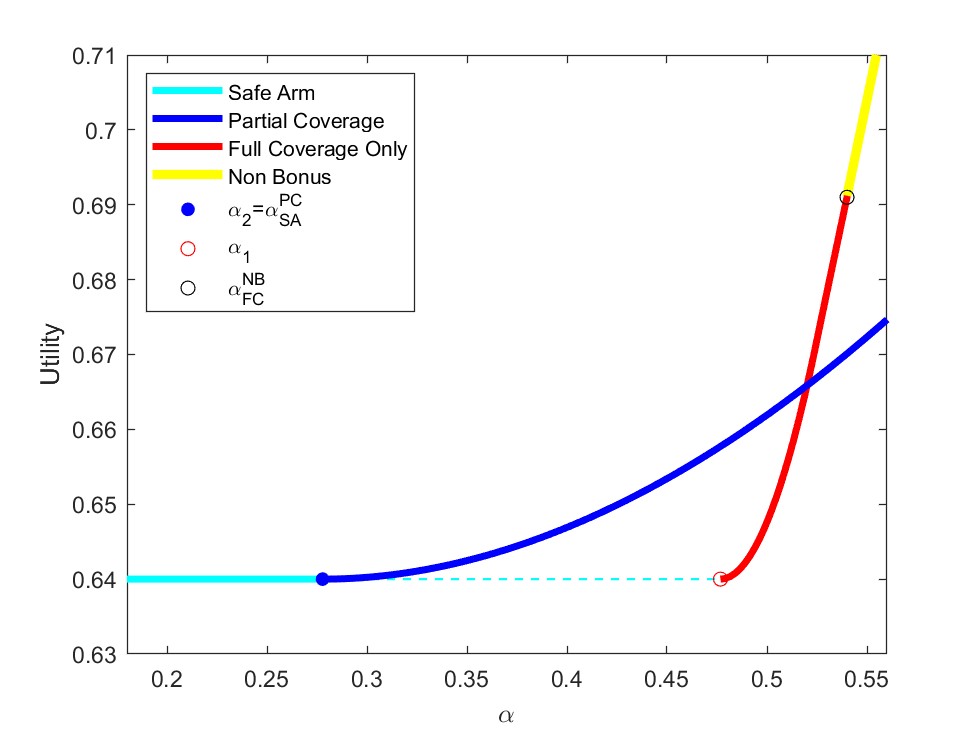}}
 	\caption{Value Functions of the Naive Monopolist}
 	\label{wrong_fc_pc}
 \end{figure}

 \subsection{Optimal Switching Monopolist}
 \label{section_optimal_switching}

In this section we consider whether the fully rational monopolist should switch from Full Coverage into Partial Coverage, or directly into Safe Arm when the belief is decreasing. From Figure \ref{wrong_fc_pc} we can see a kink for value function for Full Coverage only monopolist and Partial Coverage only monopolist. We have a claim to show it's true in general.
\begin{fact}
    The value function for $\FCO$ monopolist and  $\PCO$ monopolist always intersect.
    \label{Fact_intersect}
\end{fact}
\begin{proof}
    Obviously  from lemma \ref{fconly_lemma} we know
    $rs=r\underline\alpha_1 g+rH(\underline b_1)[\beta(\underline b_1)-\underline b_1]$, so
\[\begin{aligned}
    rs&=\underline\alpha_1gr+\underline\alpha_1\lambda(g-s)H(\underline b_1)-r\underline b_1H(\underline b_1)\\
    &<\underline\alpha_1gr+\underline\alpha_1\lambda(g-s)
\end{aligned}
\]
while from lemma \ref{pconly_lemma} we know 
\[rs=\underline\alpha_2 gr+\underline\alpha_2\lambda(g-s)\]
assumption \ref{assumption_meaningful} guarantees $g>s$, hence we always have
	\[\underline \alpha_1>\underline \alpha_2\]
 
Combined with convexity, we know $V(\alpha)$ and $U(\alpha)$ is single-crossing in non trivial part.
\end{proof}

 In other words, there should be an optimal switching between FC and PC. And PC should optimally switching to Safe Arm. Hence we denote
\[\alpha^{PC}_{SA}=\underline \alpha_2\]
Which shows that the ODE (together with boundary condition) described in lemma \ref{pconly_lemma} pin down the value function of optimal switching monopolist using PC. 

We denote the switching belief as $\alpha^{PC}_{FC}$. Then the optimal switching condition (boundary condition for the ODE described in lemma \ref{fconly_lemma} ) should be changed as
\begin{equation}
V(\alpha_{PC}^{FC})=U(\alpha_{PC}^{FC}),V'(\alpha_{PC}^{FC})=U'(\alpha_{PC}^{FC})
	\label{smooth past}
\end{equation}

 An intuitive explanation may be that Partial Coverage here can be understood as a mixed strategy of using Full Coverage and selling safe arm. A numerical illustration can be found in comparison between figure \ref{wrong_fc_pc} and figure \ref{fig:correct_fcl_pcl}, where it looks as if we try to connect $V(\alpha)$ and $U(\alpha)$ smoothly when the first optimal switching occurs. More over we can prove:
\[\alpha_{PC}^{FC}=\frac{s}{z\lambda}\]
where the right hand side is nothing but the cut-off of the myopic monopolist (i.e. the monopolist who can only consider the present profit instead of the value function including discounted future profits). Base on Lemma \ref{fconly_lemma}, \ref{pconly_lemma} and equation \ref{smooth past}, we use the following proposition to state the optimal strategy of the monopolist completely.

Moreover, we can also prove $\alpha_{FC}^{NB}>\alpha_{PC}^{FC}$, since we have $b>0$ at $\alpha_{PC}^{FC}$, and $b=0$ at $\alpha_{FC}^{NB}$.
\begin{proposition}
There exists cutoffs $0<\alpha_{SA}^{PC}<\alpha_{PC}^{FC}<\alpha_{FC}^{NB}<1$, where $\alpha_{SA}^{PC}=\frac{sr}{\lambda(g-s)+rg},\alpha_{PC}^{FC}=\frac{s}{\lambda z}, \text{ and } \alpha_{PC}^{FC} \text{ solves } V(x)=xz\lambda $, such that:
\begin{itemize}
     \item if $0<\alpha< \alpha_{SA}^{PC}$, the social planner would use Safe Arm.
     \item if $\alpha_{SA}^{PC}<\alpha<\alpha_{PC}^{FC}$, the social planner would use Partial Coverage.
     \item if $\alpha_{PC}^{FC}<\alpha<\alpha_{FC}^{NB}$, the social planner would use Full Coverage.
     \item if $\alpha_{FC}^{NB}<\alpha<1$, the social planner would use Non Bonus.
 \end{itemize}
\label{optimal_switching_proposition}
\end{proposition}
To illustrate our results, a numeric example is provided whose details can be found in Appendix \ref{appendix_numeric}.  Figure \ref{fig:correct_fcl_pcl} is a numeric result based on the same parameters in Figure \ref{wrong_fc_pc}.  The horizontal axis stands for $\alpha$ and the vertical axis stands for the value function. 


\begin{figure}[H]
\centering  
\includegraphics[width=0.5\textwidth]{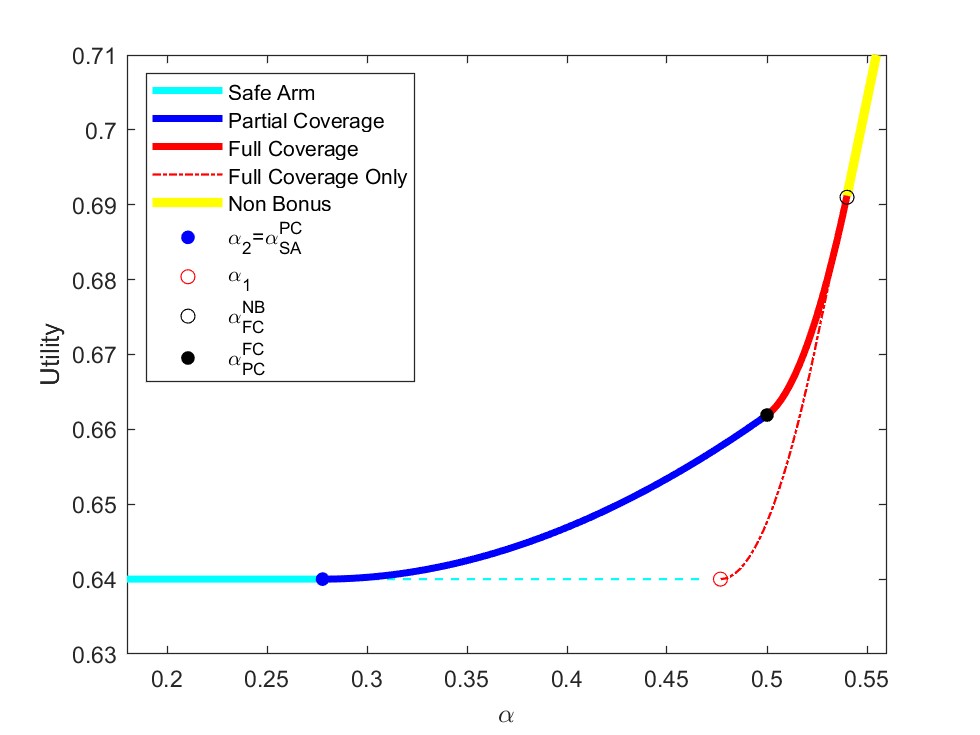}
\caption{Optimal Switching Monopolist}
\label{fig:correct_fcl_pcl}
\end{figure}

\subsection{Welfare Analysis}
\label{section_welfare}
\subsubsection{Social planner's problem : a mechanism design way}
We consider such a setting, where a social planner, commits an allocation rule $p(\alpha,\tilde c)$, a transfer rule $\tilde t(\alpha,\tilde c)$ (or equivalently, denote $t(\alpha,\tilde c)=\lim\limits_{\Delta\to 0 }\frac{\tilde t(\alpha,\tilde c)}{\Delta}$) and a reporting rule $q(\alpha,\tilde c)$, where $\tilde c$ is current agent $i$'s reported value of $c_i$, might be equal to $c_i$ or not.  If the current belief is $\alpha$, and current agent with true cost $c_i$ arrives and reports $\tilde c_i$, the social planner will allocate him a risky arm with probability $p(\alpha,\tilde c)$ with a transfer $t(\alpha,\tilde c)$, and then, conditional on  getting the risky arm, the agent has probability $q(\alpha,\tilde c)$ to be obliged to report the true utility and $1-q(\alpha,\tilde c)$ to do nothing. Reporting, allocation and transfer will be realised at the beginning of period $[t,t+\D)$, and the report (if any) will happen at the end.
Hence for the agent with true cost $c_i$, reporting $\tilde c$ leads to utility\footnote{In duration $\Delta$, the agents' utility can be written as $\mathbb E(Surplus)=-\tilde t(\alpha,\tilde c)+p(\alpha,\tilde c) \int_0^{\Delta} \lambda e^{-\lambda t}e^{-rt}\alpha z dt-p(\alpha,\tilde c) q(\alpha,\tilde c) \int_0^{\Delta}e^{-rt} c dt+[1-p(\alpha,\tilde c) ] \int_0^{\Delta} s e^{-rt} dt$ . We denote $U(c,\tilde c,\alpha)=\lim\limits_{\Delta\to 0 }\frac{\mathbb E(Surplus)}{\Delta}$ . }
\[U(c,\tilde c,\alpha)=-t(\alpha,\tilde c)+p(\alpha,\tilde c)[\alpha z \lambda - q(\alpha,\tilde c) c]+(1-p(\alpha,\tilde c))s\]
and IC constraint hence requires 
\[c\in \arg\max_{\tilde c} U(c,\tilde c,\alpha)\]

For the social planner, in period $\Delta$, if the current agent reports $\tilde c$ (we use $c$ to replace $\tilde c$ since in equilibrium we have truthfully reporting), his immediate gain from total social surplus, should be 
\begin{equation}
\tilde \pi(\triangle,\alpha,c)=p(\alpha, c) \int_0^{\Delta} \lambda e^{-\lambda t} e^{-rt} \alpha z dt -p(\alpha, c)q(\alpha, c) \int_0^{\Delta} e^{-rt}c dt +\int_0^{\Delta}e^{-rt}[1-p(\alpha, c)]sdt\label{equation_social}
\end{equation}

We define the contingent value function conditional on current $\alpha, c$ as $\hat W(\alpha,c)$, and then define the value function as $W(\alpha)$, which is based on $\alpha$ only and stands for the expected total utility when holding belief $\alpha$. Then we have (the formation details can be found in Appendix \ref{Section_formation_social_planner})
\[W(\alpha)=\mathbb E_c[\hat W(\alpha,c)]=\int_c \hat W(\alpha,c)dH(c)\]

Then we apply the operator $\int_c \cdot dH(c)$ for equation (\ref{equation_social}), (i.e. take the expectation with respect to $c$) for both sides, we can execute the limit calculation by using L'Hospital's Rule:
\[[r+\lambda \alpha \mathbb E(pq)]W(\alpha)=r\lambda \alpha z \mathbb E(p)+[1-\mathbb E(p)] sr-r\mathbb E(pqc)+\lambda^2 z\alpha \mathbb E(pq)+(\alpha^2-\alpha)\lambda W'(\alpha)\mathbb E(pq)\]

And the social planner should set the optimal $(t,q,p)$ to maximize $W(\alpha)$. Proposition \ref{proposition_socialplanner} describes the optimal strategy for the social planner. Denote $B(\alpha)\equiv \frac{\alpha\lambda^2 z}{r}+(\alpha^2-\alpha)\lambda W'(\alpha)-\lambda \alpha W(\alpha)$, $C_1(\alpha)\equiv\frac{B(\alpha)+\lambda z \alpha -s}{\lambda}$ and $C_2(\alpha)\equiv\frac{B(\alpha)}{\lambda}$, then we have: 
\begin{proposition}
    There exists optimal mechanism $(p,q,t)$ to guarantee that every agent will truthfully report their cost. There exists cutoff $\bar \alpha_{SA}^{PC}$, $\bar \alpha_{PC}^{FC}$ and $\bar \alpha_{FC}^{NB}$ and cutoff $C_1(\alpha)$ and $C_2(\alpha)$ for agents' cost,  such that:
 \begin{itemize}
     \item if $0<\alpha<\bar \alpha_{SA}^{PC}$, the social planner would set $p=q=0$ (i.e. using Safe Arm).
     \item if $\bar \alpha_{SA}^{PC}<\alpha<\bar \alpha_{PC}^{FC}$, the social planner would set $q=1$ for sure and $p=1$ if $c<C_1(\alpha)$ (i.e. using Partial Coverage).
     \item if $\bar \alpha_{PC}^{FC}<\alpha<\bar \alpha_{FC}^{NB}$, the social planner would set $p=1$ for sure and $q=1$ if $c<C_2(\alpha)$ (i.e. using Full Coverage).
     \item if $\bar \alpha_{FC}^{NB}<\alpha<1$, the social planner would set $p=1,q=0$ for sure (i.e. using Non Bonus).
 \end{itemize}
 \label{proposition_socialplanner}
\end{proposition}

Above proposition \ref{proposition_socialplanner} shows that social planner's decision is also based on the same optimal switching.  Like before, the social planner is more likely to use Full Coverage when the common belief is high (and Non Bonus is the common belief is even higher), use Safe Arm when belief is low, and use Partial Coverage as a transition when belief is moderate. More specifically, when $\alpha>\bar \alpha_{PC}^{FC}$, $p=1$ for sure, which means the social planner will always ask the agent to try the risky arm and report if the current agent is with a low enough $c$ such that $q=1$, which is like using Full Coverage.  Otherwise when $\bar \alpha_{SA}^{PC}<\alpha<\bar \alpha_{PC}^{FC}$, $p$ is possible to be zero when $c$ is large. The social planner will only ask agents with low enough costs to test the risky arm, hence it's equivalent to Partial Coverage. What’s still important here is that the once the social planner ask an agent to test the risky arm, the agent will be obliged to report(for more details please check Appendix \ref{Proof_proposition_socialplanner}). In a more intuitive way, these agents whose reporting cost is small enough to let $p=1$ for them are asked to use risky arm to explore rather than exploit the risky arm since myopically safe arm is better. When $\alpha<\bar \alpha_{SA}^{PC}$, the social planner will set $p=0$ for sure, which means the information aggregation is terminated, and the social planner sets Safe Arm. As before we can still characterize $W(\alpha)$ via a ODE system, which can be found in section \ref{proof_social_planner_optimization}.

A key message is that the switching cutoff belief is irrelevant with the decision maker. To be more specific:
\begin{enumerate}
	\item The social planner will switch to Partial Coverage at the same belief cutoff as the monopolist. At the cutoff belief, the decision maker should feel indifferent myopically between safe arm and risky arm, since the reporting probability is continuous at this cutoff. 
	\item The social planner will stop at the same belief cutoff as the monopolist. Though surprising at the first sight, actually the switching condition guarantees this equivalence. $\underline \alpha_2$ actually stands for the switching belief between Partial Coverage and Non Bonus, at which in the long run the decision maker feels indifferent between doing more experiment or not. One more experiment is of no value when the common belief is low enough, and it's irrelevant with the cost of experiment. Intuitively speaking, because of the existence of Partial Coverage, if the information is valuable for the social planner, it's also valuable for the monopolist since the monopolist can always set a low enough bonus and use Partial Coverage, and keep waiting until agents with low enough cost arrive, try the risky arm and report. 
\end{enumerate}

\subsubsection{Comparison of social planner's strategy and monopolist's strategy}
\label{section_comparison}
We measure the efficiency of the learning process led by the monopolist in two ways. Firstly, will the monopolist stop experiments at a higher belief? And secondly, will the monopolist attach the termination sooner than the social planner? The following proposition answers the first question:

\begin{proposition}
	  At any belief, the social planner will do experiments with (weakly) higher probability. As a result  $\bar \alpha_{FC}^{NB}> \alpha_{FC}^{NB}$. Moreover, the social planner will switch from Full Coverage to Partial Coverage, and switch from Partial Coverage and Non Bonus at the same believes as the monopolist (i.e. $\bar \alpha_{SA}^{PC}= \alpha_{SA}^{PC}$, and $\bar \alpha_{PC}^{FC} = \alpha_{PC}^{FC}$).
   \label{proposition_speed}
\end{proposition}

Proposition \ref{proposition_speed} that the monopolist's willingness do experiments is less than the social planner, hence every stage there is lower probability to report. Since the social planner bears less cost than the monopolist ($\mathbb E(c|c<b)<b$), a quick thought should be that the social planner will encourage agents with higher reporting cost to report than the monopolist. In other words, the learning speed of the social planner should be higher. ``Learning speed'' actually means reporting probability in our model (instead of arrival rate $\lambda$). We should realize that here social planner can naturally do better since he can design the reporting probability on $\alpha,c$ while the monopolist can only set it on $\alpha$. 
\begin{figure}[H]
\centering  
\includegraphics[width=0.5\textwidth]{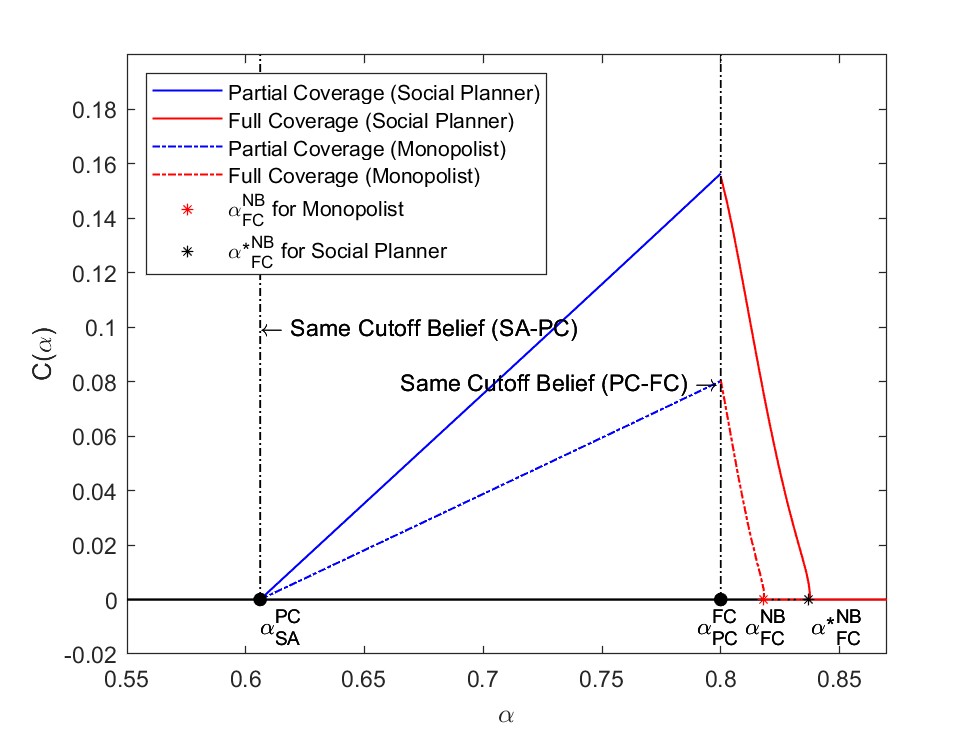}
\caption{Mechanisms Comparison }
\label{fig:pqplot}
\end{figure}

A further illustration is in Figure \ref{fig:pqplot}. Both using Full Coverage, when the current agent's reporting cost $c$ is lower than the red solid line, the social planner is willing to set $p=q=1$, while the monopolist is only willing to set $p=q=1$ when  $c$ is lower than the red dotted line. We can find the monopolist will have lower probability to do experiments than the social planner.
Similarly, blue lines stands for the Partial Coverage, and our observations still hold.

\subsubsection{Social Surplus under Monopolist's Setting}
Section \ref{section_comparison} demonstrates that the learning speed under the monopolist setting is consistently lower than that under the social planner. A subsequent question arises: does the slower learning speed lead to inefficiency? To address this inquiry, we require the value function of social surplus under the monopolist setting. Intuitively, the monopolist sets a lower bonus primarily to avoid incurring reporting costs. Hence, a natural conjecture is that a lower learning speed aligns with inefficiency. In this section, we will substantiate this intuition (refer to Figure \ref{fig:newsocial}).

To calculate the social surplus under the monopolist setting, we first express the value function and payoff contingent on the current agent's cost $c$ and then take expectations (More details are in Appendix \ref{Social_Surplus_Monopolist}). Conditional on $c$, we denote $P_L(b)=\mathbb I(b>c)$ as whether the current agent would report and modify equation \eqref{equation_pi_1l} as
\begin{equation}
    \mathbb E_c\lim_{\Delta\to0}\frac{\pi_{L}(\Delta,c)}{\Delta}=\left\{
\begin{aligned}
	&\alpha z \lambda - H(b)\mathbb E(c|c<b)& \text{ if } L=\text{FC}\\
	&H(b) \lambda z \alpha-H(b)\mathbb E(c|c<b) +[1-H(b)]s& \text{ if } L=\text{PC}\\
	& s&\text{ if } L=\text{SA}\\
    & \alpha z  \lambda &\text{ if } L=\text{NB}
\end{aligned} \right.
\label{stage_value_function_social_surplus}
\end{equation}

\begin{figure}[H]
	\centering
	\includegraphics[width=0.7\linewidth]{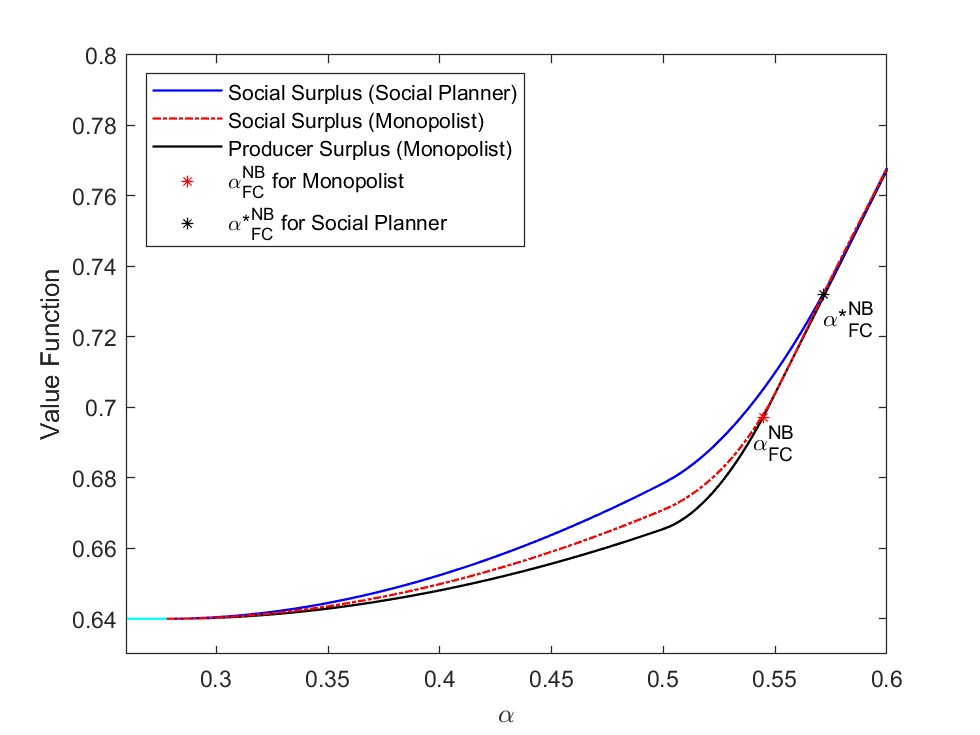}
	\caption{Social Welfare Comparison}
	\label{fig:newsocial}
\end{figure}

Under the monopolist setting, the disparity between social surplus and producer surplus arises because the seller incurs a loss of $b_t>c_t$ when successfully motivating agent $t$ to report, whereas for the social surplus, the loss is $c_t$. As a natural consequence, the value function of social surplus is always weakly higher than the value function of producer surplus. Furthermore, they coincide when there is no reporting, i.e., when $\alpha>\alpha_{FC}^{NB}$ or $\alpha<\alpha_{SA}^{PC}$. In Figure \ref{fig:newsocial}, $\alpha_{FC}^{NB}$ represents the monopolist switching cutoff, at which the red dotted line intersects with the black solid line.

This discrepancy also extends to the social planner's setting. Thus, if the social planner optimally designs the mechanism to maximize social surplus, they can outperform the monopolist, resulting in a value function that is also weakly higher than the value function of social surplus under the monopolist setting. As illustrated in Figure \ref{fig:pqplot}, the social planner switches between Partial Coverage and Safe Arm at the same cutoff as the monopolist but shifts to Full Coverage from Non Bonus at a higher cutoff. We use $\alpha_{FC}^{*NB}$ to denote the social planner's switching cutoff. For $\alpha>\alpha_{FC}^{*NB}$—since both the social planner and monopolist don't conduct experiments—their value functions should be the same. That's why in Figure \ref{fig:newsocial} both the black line and red dotted line intersect with the blue solid line when $\alpha>\alpha_{FC}^{*NB}$. 

In conclusion, this comparison reveals that social surplus under the monopolist setting is weakly lower than that under the social planner setting, suggesting that the monopolist would lead an inefficient learning process. Even when both learning processes terminate at the same cutoff, a lower learning speed adversely affects social surplus.

\subsection{Comparative Statics}
Here we do some comparative statics. We have discussed the effects of monopolist's behavior caused by discount rate $r$, relative utility $R_2$ , $\mathbb ER_1$ and $\mathbb \max \mathbb E\{R_1.R_2\}$. Here we add some new intuitions: how $\lambda$ (the perfect level of experiment) affects the monopolist's behavior in section \ref{comparative_lambda}, and how discount rate affects the social planner's behavior in section \ref{comparative_r}.  

\subsubsection{Comparative Statics for $\lambda$}
\label{comparative_lambda}
In this case we focus on discussing how the arrival rate $\lambda$ affects our results. Actually the larger arrival rate stands not only the risky arm is more attractive when the state is known to be good (since selling the risky arm leads to immediate profit $\alpha \lambda z$ ), but also the measure of the imperfect learning (since after purchasing the risky arm, there is probability $\alpha(1-e^{-\lambda \Delta})$ to get a positive feedback after an experiment with duration $\Delta$). So if we fix $s$, a larger $\lambda$ will make FC and NB more attractive hence we will see $\alpha_{SA}^{PC}$ and 
$\alpha_{PC}^{FC}$ will shrink to zero. A more meaningful way is to focus on the second power: how will the results change if the learning is ``less imperfect''. To do so we can re-define $s=\lambda z \rho$ where $\rho<1$ is a fixed constant who guarantees that the safe arm and risky arm is comparable in stage utility (i.e. to rule out the first effect). The intuition should be that if the learning is ``less imperfect'', then the seller will be more willing to do the experiment, hence we should expect a smaller $\alpha_{SA}^{PC}$ (and tends to 0 if $\lambda \to \infty$), and a larger $\alpha_{SA}^{PC}$ (and tends to 1 if $\lambda \to \infty$). The first intuition is easy to check since we have the analytical solution for $\alpha_{SA}^{PC}$. As for the $\alpha_{FC}^{NB}$, we have the following illustration Figure \ref{fig:lambda} by letting 
$\rho=0.5,r=0.5,z=7$ and we let $\lambda$ to grow from 0.2 to 2.

\begin{figure}[H]    
  \centering           
  \subfloat[$\lambda=0.3$]  
  {
\label{fig:subfig1}\includegraphics[width=0.4\textwidth]{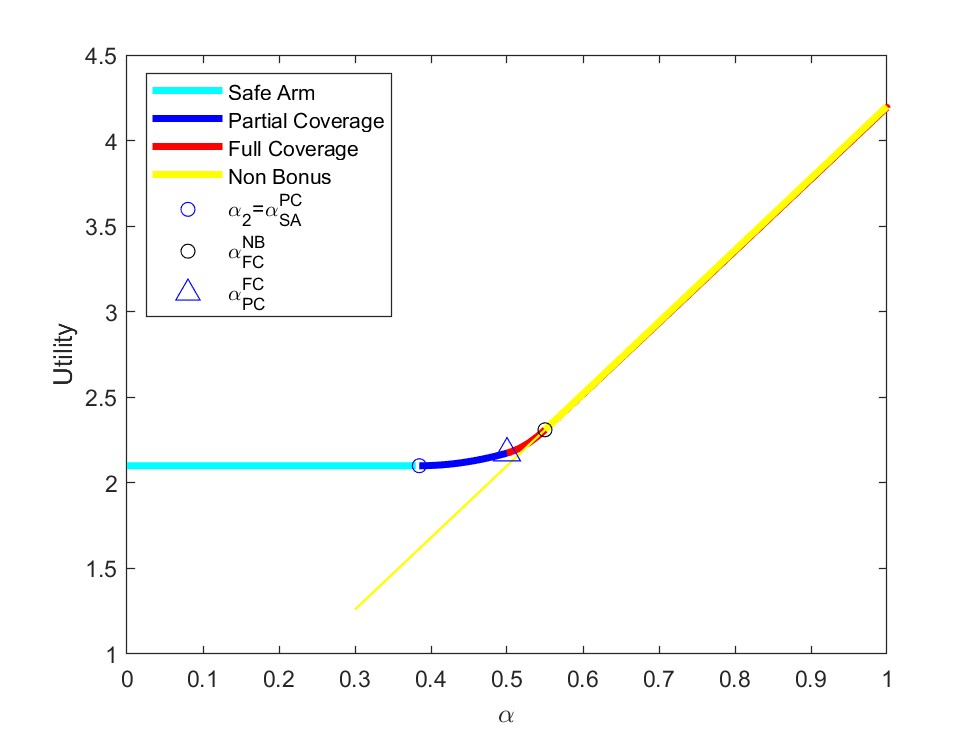}
  }
  \subfloat[$\lambda=0.8$]
  {
\label{fig:subfig2}\includegraphics[width=0.4\textwidth]{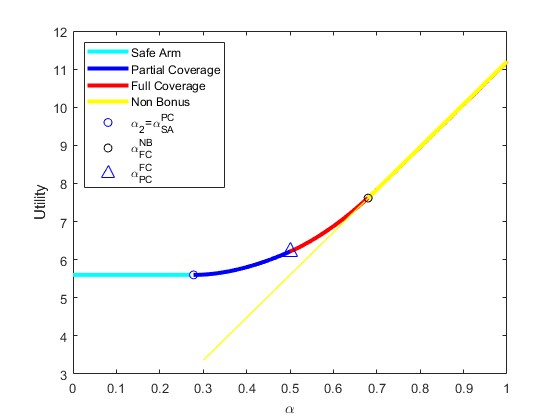}
  }
  
  \subfloat[$\lambda=1$]
  {
\label{fig:subfig3}\includegraphics[width=0.4\textwidth]{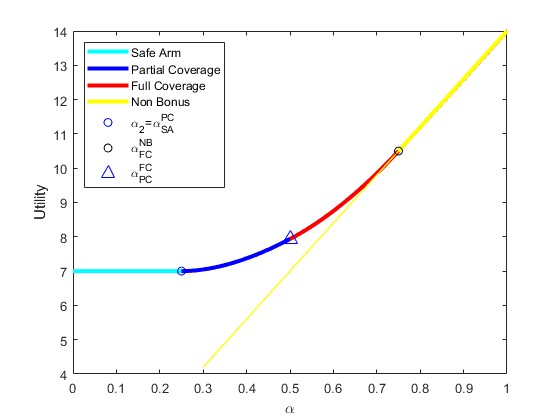}
  }
  \subfloat[$\lambda=2$]
  {
\label{fig:subfig2}\includegraphics[width=0.4\textwidth]{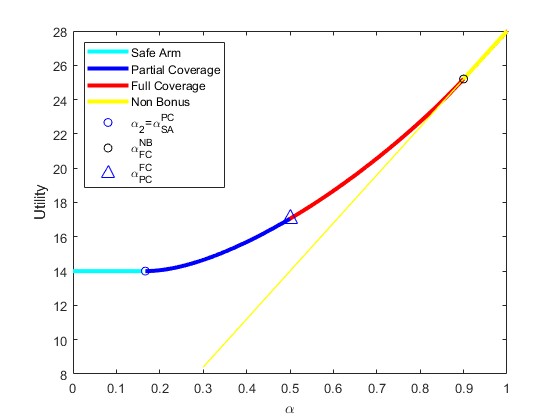}
  }
  \caption{$\lambda$ Affects $\alpha_{FC}^{NB}$}    
  \label{fig:lambda}            
\end{figure}

\subsubsection{Comparative Statics for $r$}
\label{comparative_r}
We do comparative statics here with respect to  $r$, the discount factor of the decision maker. A direct idea is value functions should be monotonically decreasing with respect to the discount rate $r$. Moreover, with higher $r$, the decision maker will be less willing to do experiments, no matter in Monopolist case or Social Planner case. An illustration (take social surplus as an example) is like Figure \ref{fig:r}:

\begin{figure}[H]    
  \centering            
  \subfloat[$r=0.2$]  
  {
\label{fig:subfig1}\includegraphics[width=0.4\textwidth]{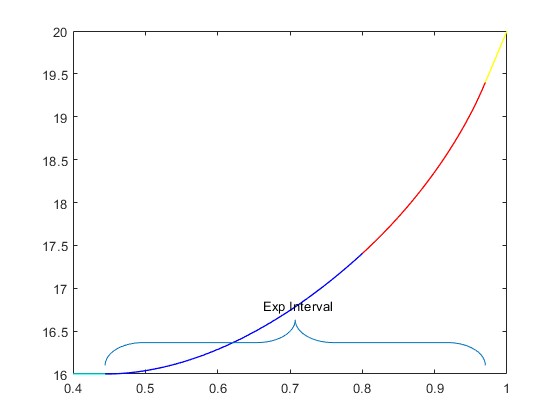}
  }
  \subfloat[$r=0.5$]
  {
\label{fig:subfig2}\includegraphics[width=0.4\textwidth]{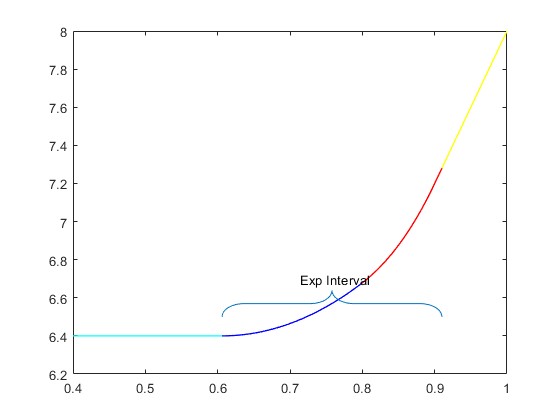}
  }
  
  \subfloat[$r=0.9$]
  {
\label{fig:subfig3}\includegraphics[width=0.4\textwidth]{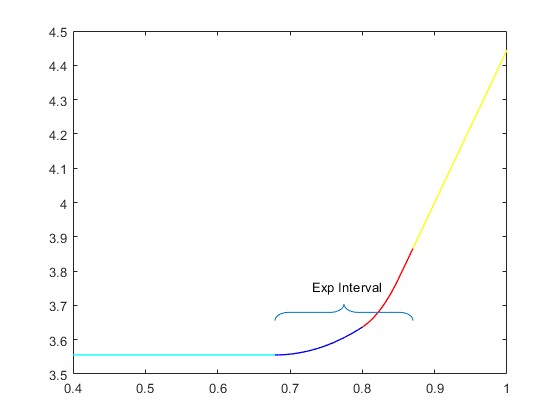}
  }
  \subfloat[$r=1.2$]
  {
\label{fig:subfig2}\includegraphics[width=0.4\textwidth]{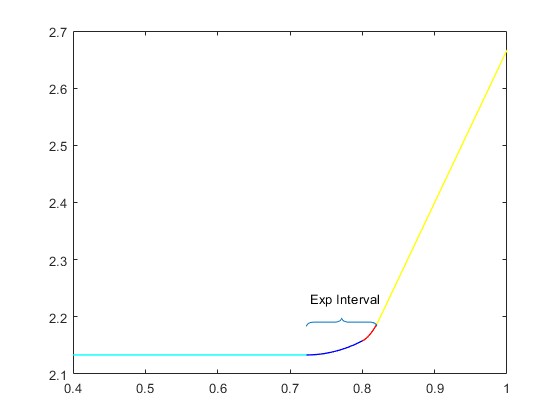}
  }
  \caption{$r$ Affects Experiment Incentive}    
  \label{fig:r}            
\end{figure}

$r$ here stands for the patience of the decision maker, hence more efficient experiment will be allowed when he is more patient. The decision whether to switch from Full Coverage to Partial Coverage seems like a myopic decision since $\underline \alpha_3$ is not affected by $r$: the monopolist stops selling risky arm and starts to sell safe arm when risky arm makes less profit. However, it's actually not the case. As we have mentioned before, Partial Coverage has the same learning speed with Full Coverage. Switching to Partial Coverage will not slow down the speed of learning the state, hence $r$ will not affect the switching belief here.
\newpage
\section{Immediate Revelation}

In Section \label{section_Imperfect_Learning}, we set up our model enforcing  \ref{assumption_barc_large} which allows us to ignore the potential Immediate Revelation strategy in Section \ref{section_Model_Analysis}. In this section we loose this assumption, and consider the case where $\bar c\to 0$. In this case, agents will report the utility of risky arm for sure after purchasing the risky arm. 

\label{section_IR}
\subsection{Bonus and Value Function}
\label{section_IR_value_function}
If Immediate Revelation is a potential strategy, then using this strategy, the seller would encourage the current agent to purchase the risky arm and report for sure. Hence it can be treated as a corner solution for both the Partial Coverage or Full Coverage strategy by setting $b_t=\bar c$. For stage payoff, we have
$\lim\limits_{\D \to 0}\frac{\pi_L(\D,b_t=\bar c)}{\D}=\lambda z \alpha -\bar c$ if $L=$IR. Since we have shown that the bonus structure is single peaked , monotonically increasing with respect to $\alpha$ when $\alpha<\alpha_{PC}^{FC}$ and monotonically decreasing with respect to $\alpha$ when $\alpha>\alpha_{PC}^{FC}$, the only possible solution is $\exists \alpha_{PC}^{IR}\le \alpha_{PC}^{FC}\le\alpha_{IR}^{FC}$, such that for $\alpha\in (\alpha_{PC}^{IR},\alpha_{IR}^{FC})$ the seller would optimally choose IR. In other words, the only optimal switching rule is 
\[SA\rightarrow PC\rightarrow (IR) \rightarrow FC \rightarrow NB\]
An illustration is Figure \ref{fig:IRpqplot}. Compare to Figure \ref{fig:pqplot}, $z$ is larger, so relatively the seller has more incentive to reveal the good state sooner. As a result, the bonus scheme will uniformly increase, making the upper bound $\bar c$ is possible to be tight. Hence there is an upper truncation in bonus scheme. The horizontal part stands for using Immediate Revelation.  
\begin{figure}[H]
	\centering
	\includegraphics[width=0.7\linewidth]{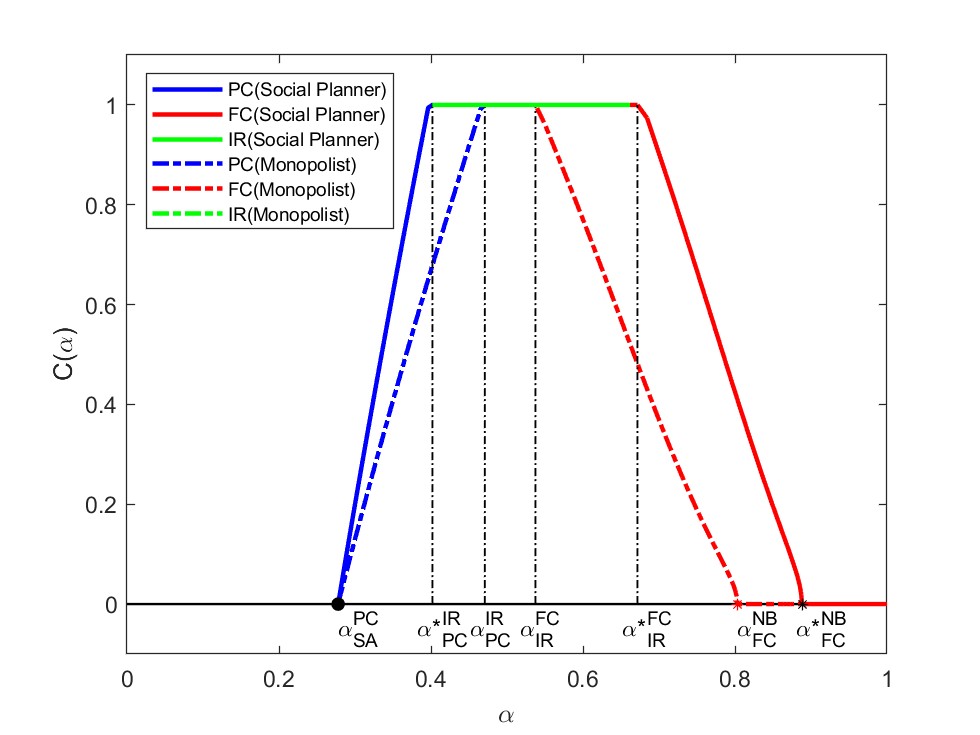}
	\caption{Mechanism Comparison with IR}
	\label{fig:IRpqplot}
\end{figure}
Now we can characterize the value function of seller. The law of motion like \eqref{result ode} also holds\footnote{Notice here if we allow this $Y(\cdot)$ to optimally switch to $\tilde U$, i.e., $\exists \tilde \alpha$ such that $Y(\tilde \alpha)=\tilde U=\frac{s}{r}$, and $Y'(\tilde \alpha)=0$, then the cutoff belief should satisfy $\tilde \alpha=\frac{r(s+\bar c)}{rg+\lambda (g-s)}$. A sufficient condition to ignore Immediate Revelation is $\tilde \alpha\ge 1$, which is our Assumption \ref{assumption_barc_large}.}: 
 \begin{equation}
    r Y(\alpha)= \lambda \alpha z(1+ \frac{\lambda}{r})-\bar c-\alpha \lambda Y(\alpha)+(\alpha^2-\alpha)\lambda Y'(\alpha)
    \label{IR_ode}
\end{equation}
with smooth pasting condition
\[U(\alpha_{PC}^{IR})=Y(\alpha_{PC}^{IR})
,U'(\alpha_{PC}^{IR})=Y'(\alpha_{PC}^{IR})\]
and
\[Y(\alpha_{IR}^{FC})=V(\alpha_{IR}^{FC}),Y'(\alpha_{IR}^{FC})=V'(\alpha_{IR}^{FC})\]
Compare to Figure \ref{fig:correct_fcl_pcl}
\begin{figure}[H]
	\centering
	\includegraphics[width=0.7\linewidth]{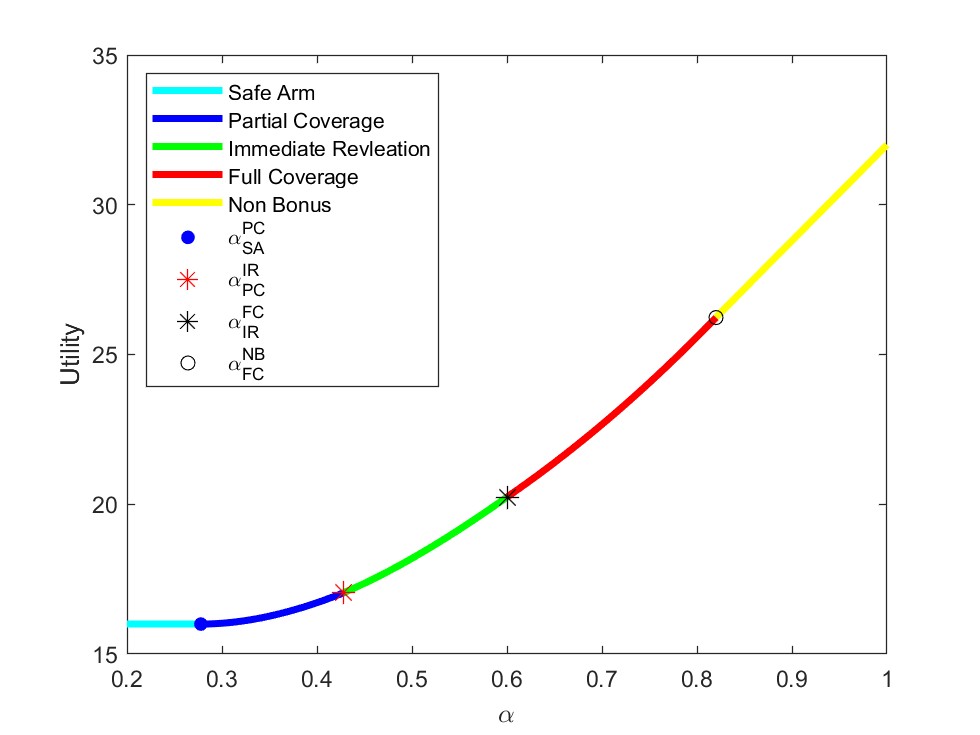}
	\caption{Optimal Switching Monopolist with IR}
	\label{fig:IRvalue_function}
\end{figure}
More details about the numerical example can be found in Appendix \ref{section_numerical_IR}.

\subsection{Comparison with \cite{bergemann2000experimentation}}
\cite{bergemann2000experimentation}'s  model can be treated as an extreme case where $\bar c\to 0$. Since in their model risky arm and safe arm are sold by two different sellers, only the social efficient allocation is comparable. Their main conclusion is, there exists a social efficient stopping level $\hat \alpha$ at which risky arm is myopically dominated by the safe arm. When $\alpha\in (\hat \alpha,1)$, the social planner should always encourage experiment on risky arm. Moreover, because of the Bertrand competition nature in their model, the risky product seller may lead to excess learning. In our model, we general the first part of the result: if there is reporting cost with support $[0,\bar c]$, the first best allocation is, under moderate belief level $[\alpha_{SA}^{PC},\alpha_{FC}^{NB}]$ to \textbf{screen out} the agent with lower reporting cost to do experiment on risky arm and report while allocate the other agent with myopically dominant one. If belief is extremely low, then with the same logic as in \cite{bergemann2000experimentation}, the society should give up exploration on the risky arm. If belief is extremely high, exploration would generate marginal social surplus compared to the reporting cost, hence the first best allocation is to set Non Bonus (i.e., let all agents buy risky arm without learning). With $\bar c\to 0$, $\alpha_{FC}^{NB} \to 1$ and $\alpha_{SA}^{PC}\to \hat \alpha$.

In their setting, agents are long lived with utility value funtion 
\[d u_i(t)=\mu d t+\sigma d W_i(t)\]
Their Pareto efficient stopping cutoff is claimed in the following Theorem 1:

\textbf{Theorem 1.} \textit{(Efficient Stopping). The Pareto efficient stopping point $\hat{\alpha}$ is
$$
\hat{\alpha}=\frac{\left(s-\mu_L\right)(\gamma-1)}{\left(\mu_H-\mu_L\right)(\gamma-1)+2\left(\mu_H-s\right)},
$$
with
$$
\gamma=\sqrt{1+\frac{8 r \sigma^2}{N\left(\mu_H-\mu_L\right)^2}}
$$}

Numerically, if we set $N=1$, $\mu_L=0$, $\mu_H=g=\lambda z$, $\lambda=2r$, $\sigma^2=\frac{3g^2}{4\lambda}$, and $\gamma=2$, then $\hat{\alpha}=\alpha_{SA}^{PC}$. However, although the models are similar, they are not identical. In their model, the agent's utility function has variance $\lim\limits_{\Delta\to 0} \frac{Var(u(\Delta))}{\Delta}=\sigma^2$, which is \textbf{irrelevant} to the current belief. In our model, the utility within $\Delta$ can be considered as a Bernoulli distribution, with a probability of $1-\gamma$ to realize utility $z$ and $\gamma$ to realize 0. Hence, the variance satisfies $\lim\limits_{\Delta\to 0}\frac{Var(u_t(\Delta))}{\Delta}=\frac{\alpha g^2}{\lambda}$. A more detailed illustration is Figure \ref{pqplot_multibarc}, where  Figure \ref{subfig1} shows that when $\bar c\to0$, the monopolist's bonus would shrink to a cutoff strategy like $\epsilon \to0^{+}, \forall \alpha>\alpha_{SA}^{PC}$, with the strategy defined in Section \ref{section_Model_Analysis}, the monopolist would allocate all market demand to risky arm when $\alpha>\alpha_{SA}^{PC}$, otherwise safe arm. Figure \ref{subfig2} shows the same logic works for the social planner. More detailed discussion can be found in online appendix. 
\begin{figure}[H]    
  \centering            
  \subfloat[Monopolist]  
  {
\label{subfig1}\includegraphics[width=0.4\textwidth]{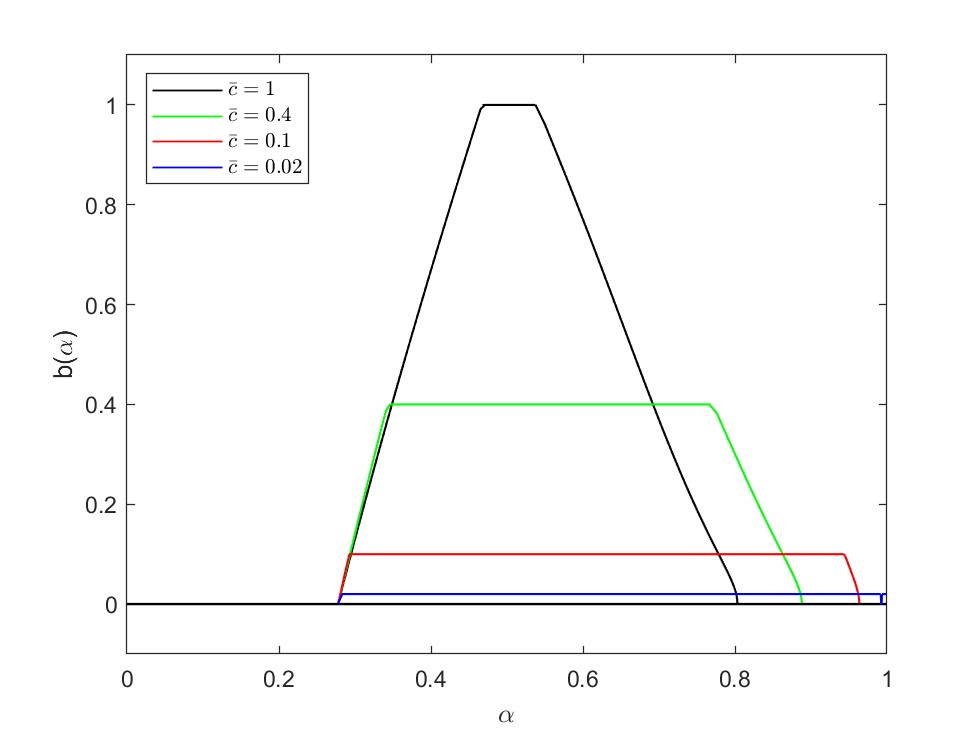}
  }
  \subfloat[Social Planner]
  {
\label{subfig2}\includegraphics[width=0.4\textwidth]{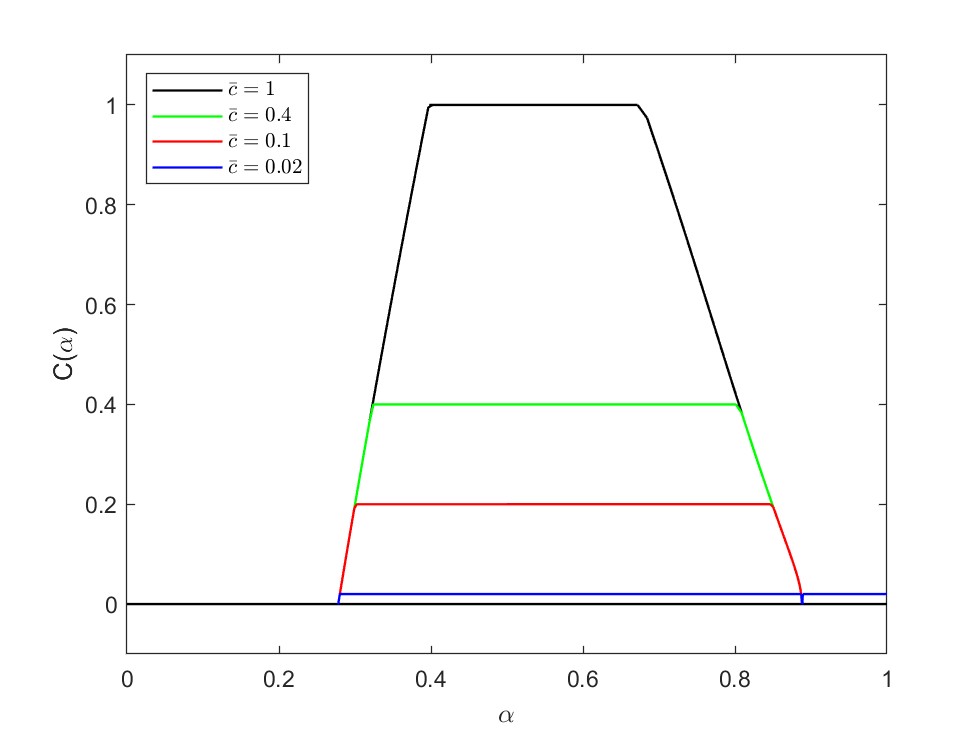}
  }       
  \caption{When $\bar c\to 0$}
  \label{pqplot_multibarc}
\end{figure}

\section{Conclusion}
The monopolist always has the ability to dynamically design pricing schemes for risky products, hence even when risky products are myopically less profitable than safe products, the monopolist has motivation to continue learning. 
If there are reporting costs for agents, the monopolist also needs to jointly design bonus schemes, which essentially allow sellers to influence the speed of learning without distorting market demand through price manipulation. When risky products are less profitable but still worth exploring, sellers can employ a Partial Coverage strategy by selling safe products to agents who won't report, thus increasing overall profits while continuing to incentivize exploration through bonus design. Rational sellers should switch from Full Coverage to Partial Coverage rather than stopping experiments immediately. As sellers hold profits from both types of products, the learning process steered by sellers terminates at a socially optimal level, yet the presence of reporting costs slows down the learning pace in each period.


\newpage

\appendix
\newpage
\section{Appendix: Proofs for Continuous Case}
\renewcommand\thefigure{\Alph{section}\arabic{figure}} 
\renewcommand\thefigure{\Alph{section}\arabic{table}} 
\setcounter{figure}{0}  

\subsection{Proof: Formation of Law of Motion in Section \ref{Section_preliminary_analysis}}
\label{appendix_proof_pre}
In the formal context we have shown that the value function $\Pi(\cdot)$ follows the dynamic process:
\begin{equation}
\Pi(\alpha)=\max\limits_{L,b,R} \pi_{1,L} (\Delta)+P_L(b,R)e^{-r\Delta}\gamma(\alpha,\Delta) \Pi(\alpha')+
	P_L(b,R) e^{-r\Delta}(1-\gamma(\alpha,\Delta)) \tilde V
	+(1-P_L(b,R))e^{-r\Delta}\Pi(\alpha)
\label{appendix_Pi}
\end{equation}

where \[\gamma(\alpha,\Delta)=1-\alpha(1-e^{-\lambda \Delta})\] 
$P_L(b)$ stands for the probability of motivating the current agent to report successfully under strategy $L$.

Combining with Bayesian updating rule
\[
\alpha'=\frac{\alpha e^{-\lambda \Delta}}{\alpha e^{-\lambda \Delta}+1-\alpha}
\]
hence 
\[
\alpha'-\alpha=\frac{(\alpha-\alpha^2) (e^{-\lambda \Delta}-1) }{\alpha e^{-\lambda \Delta}+1-\alpha}
\label{appendix_bayes}
\]
So $\alpha'-\alpha\to 0$ if $\Delta\to 0$. Moreover, we have
\[\lim_{\Delta\to 0}\frac{\alpha'-\alpha}{\Delta}=(\alpha^2-\alpha)\lambda\]
and
\[\lim_{\Delta\to 0}\frac{(\alpha'-\alpha)^2}{\Delta}=0\]

We use the Taylor expansion with Lagrangian remainder for $\Pi(\alpha)$, so there exists $\alpha\ge \xi\ge \alpha'$ such that
\[\Pi(\alpha')=\Pi(\alpha)+\Pi'(\alpha)(\alpha'-\alpha)+\frac{\Pi''(\xi)}{2}
(\alpha'-\alpha)^2\]

So the second term of equation \ref{appendix_Pi} can be expressed as
\[P_L(b,R)e^{-r\Delta}\gamma(\alpha,\Delta) [\Pi(\alpha)+\Pi'(\alpha)(\alpha'-\alpha)+\frac{\Pi''(\xi)}{2}
(\alpha'-\alpha)^2]\]
hence equation \ref{appendix_Pi} can be expressed as
\[
\begin{aligned}
	\Pi(\alpha)[1-(1-P_L(b,R))e^{-r\Delta}-P_L(b,R)e^{-r\Delta}\gamma(\alpha,\Delta)]=\\ \max\limits_{L,b,R} \pi_{1,L} (\Delta)+P_L(b,R)e^{-r\Delta}\gamma(\alpha,\Delta) [\Pi'(\alpha)(\alpha'-\alpha)+\frac{\Pi''(\xi)}{2}
	(\alpha'-\alpha)^2]+
	P_L(b,R) e^{-r\Delta}(1-\gamma(\alpha,\Delta)) \tilde V
\end{aligned}
\]
Exerting operation $\lim\limits_{\Delta\to0}\frac{\cdot}{\Delta}$ on both sides, we can get $\Pi(\alpha)$ actually follows the ordinary differential equation (here we exchange the order of $\max$ and $\lim$ for the first term in right hand side):
\[
	[r+P_L(b)\alpha \lambda ] \Pi(\alpha)=  \max\limits_{L,b,R} \lim_{\Delta\to0}\frac{\pi_{1,L} (\Delta)}{\Delta}+P_L(b)(\alpha^2-\alpha)\lambda \Pi'(\alpha)+P_L(b) \frac{\alpha \lambda^2 z }{r}
\]
With Berge's Maximum theorem, since the feasible field is tight (for $b$ and $R$ there is lower bound 0 and upper bound $\bar c$), $\Pi(\alpha)$ is continuous. 

\subsection{Proof of Lemma \ref{fconly_lemma} (Full Coverage Only Monopoly)}
\label{proof_fconly}
We have 
 \[
 rV(\alpha)=\max\limits_{b}\alpha z \lambda[1+H(b)\frac{\lambda}{r}]- H(b)b+H(b)(\alpha^2-\alpha)\lambda V'(\alpha)-H(b)\lambda\alpha V(\alpha)
 \tag{FC-ODE}
 \]
For interior solutions the following F.O.C with respect to $b$ holds:
\begin{equation}
    h (b)\lambda\alpha V(\alpha)=\alpha z \lambda h(b)\frac{\lambda}{r}- h(b)b- H(b)+h(b)(\alpha^2-\alpha)\lambda V'(\alpha)
    \label{appendix_fc_foc}
\end{equation}

F.O.C can be written as:
\begin{equation}
    \beta(b)=\frac{\alpha z \lambda^2}{r} +(\alpha^2-\alpha)\lambda V'(\alpha)-\alpha\lambda V(\alpha) 
    \label{appendix_foc_beta}
\end{equation}

Assumption \ref{assumption:monotone_beta} requires that 
$2h^2(b)-h'(b)H(b)> 0$ for $\forall b \in [0,\bar c]$. Hence to make sure the interior solution is a local maximum, we need the following S.O.C.:
\begin{equation}
    \alpha z \lambda h'(b)\frac{\lambda}{r}- h'(b)b-h(b)- h(b)+h'(b)(\alpha^2-\alpha)\lambda V'(\alpha) -h '(b)\lambda\alpha V(\alpha) \le 0
    \label{appendix_fc_soc}
\end{equation}

at the interior solution. Combined with equation \ref{appendix_fc_foc}, the equation \ref{appendix_fc_soc} is  actually equivalent to 
\[h'(b)H(b)-2 h^2(b) \le 0\]
Hence S.O.C. is satisfied. Then we prove the first part of Lemma \ref{fconly_lemma}.

Plug the equation \ref{appendix_fc_foc} into equation (FC-ODE) we can get:
\begin{equation}
   r V(\alpha)=\alpha z \lambda+ \frac{H^2(b^*)}{h(b^*)} \label{appendix_V_alpha_express} 
\end{equation}
 So we prove the second part of Lemma \ref{fconly_lemma}.

At stopping belief (we use $\underline \alpha_1$ to stand for the cutoff point here in order to make difference) we have $V(\underline{\alpha_1})=\tilde U$. Moreover, we have smooth pasting condition $V'(\underline{\alpha_1})=0$. We denote $\underline b_1$ for optimal bonus when belief is $\underline{\alpha_1}$. Hence $\underline b_1$ maximizes $V(\alpha)$ at $\underline{\alpha_1}$, plug into equation \ref{appendix_foc_beta} and \ref{appendix_V_alpha_express}, we can get:
\[
 \beta(\underline b_1)=\frac{\underline \alpha_1  \lambda ^2 z -\underline \alpha_1  \lambda s}{r}
\]
and :
\[
s =\underline \alpha_1  z \lambda+ \frac{H^2(\underline b_1)}{h(\underline b_1)}
\]

So we prove the third part of Lemma \ref{fconly_lemma}.

As for the monotonicity of $b(\alpha)$, since equation \ref{appendix_V_alpha_express} shows that
\[rV(\alpha)-\alpha z\lambda= \frac{H^2(b)}{h(b)}\] and right hand side of this equation is monotonically increasing with $b$, we only need to show the left hand side is monotonically decreasing with $\alpha$. And let $r V'(\alpha)-z\lambda=k$, from equation \ref{appendix_foc_beta} and \ref{appendix_V_alpha_express}  we have \[\frac{(\alpha^2-\alpha)\lambda k}{r}=\lambda \beta(b)+\frac{\alpha \lambda}{r} \frac{H^2(b)}{h(b)}>0\] hence we have $k<0$ and complete our proof of the forth part.

\subsection{Proof of Lemma \ref{pconly_lemma} (Partial Coverage Only Monopolist)}
The value function here fits
 \[
 r U(\alpha)=\max_b H(b) (\lambda \alpha z-b) +[1-H(b)]s+H(b)\frac{\alpha \lambda^2 z}{r}+H(b)(\alpha^2-\alpha)\lambda U'(\alpha)-H(b)\alpha \lambda U(\alpha)
 \tag{PC-ODE}
 \]
F.O.C respect to $b$ holds:
\[
    h(b)\alpha \lambda U(\alpha)=h(b) (\lambda \alpha z-b)-H(b) -h(b)s+h(b)(\alpha^2-\alpha)\lambda U'(\alpha)+h(b)\frac{\alpha \lambda^2 z}{r}
\]
which can be written as :
\[
\beta(b)=\lambda \alpha z - s+\frac{\alpha z \lambda^2}{r}+(\alpha^2-\alpha)\lambda U'(\alpha)-\lambda \alpha U(\alpha)
\]
Similar as in section \ref{proof_fconly}, S.O.C. is satisfied hence the interior solution here maximizes the value function. Hence we prove the first part of Lemma \ref{pconly_lemma}.
By plugging the F.O.C into expression of $U(\alpha)$ we can get Bellman Equation and prove the second part of Lemma \ref{pconly_lemma}.:
\begin{equation}
    rU(\alpha)=s+\frac{H^2(b)}{h(b)} 
    \label{appendix_pc_BE}
\end{equation}

Still, at stopping belief $\underline \alpha_2$ we have the following boundary condition:
\[
U(\underline \alpha_2)=\Tilde U,U'(\underline \alpha_2)=0 
\]
 hence
\[
\underline \alpha_2=\frac{sr}{\lambda (g-s)+gr}
\]
and 
\[
\underline b_2= 0 
\] 
So the third part is proven. As for the fourth part, actually $b$ shares the same monotonicity as $U(\alpha)$ from equation \ref{appendix_pc_BE}. So the proof is completed. 

\subsection{Proof of Proposition \ref{optimal_switching_proposition} (Optimal Switching Monopolist)}
To finish the proof of proposition 5.3 we only need to show that $\underline \alpha_{PC}^{FC}=\frac{s}{z\lambda}$.
We denote $\underline \alpha_3 = \underline \alpha_{PC}^{FC}$, and $\underline b_3$, $\underline R_3$ for the corresponding optimal $b$ or $R$.
Value these equation from  Lemma \ref{fconly_lemma} and \ref{pconly_lemma} at $\underline \alpha_{PC}^{FC}$ and combined with our smooth pasting condition
\[U(\underline \alpha_3)=V(\underline \alpha_3),U'(\underline \alpha_3)=V'(\underline \alpha_3)\]we have: 
 \[\frac{s-\underline \alpha_3z\lambda}{\lambda}=\beta(\underline b_3)-\beta(\underline R_3)=\frac{H^2(\underline b_3)}{h(\underline b_3)}-\frac{H^2(\underline R_3)}{h(\underline R_3)}\]
Denote $Q(x)=\frac{H^2(x)}{h(x)}-\beta(x)$, we have (from our assumption 3.2, $2h^2(x)-H(x)h'(x)>0$): \[Q'(x)=\frac{2h^2(x)-H(x)h'(x)}{h^(x)}[H(x)-1]<0\] 
Hence we get $\underline b_3=\underline R_3$, and $\underline \alpha_3=\frac{s}{z\lambda}$.

\subsection{Proof of Proposition \ref{proposition_socialplanner}(Welfare Analysis)}
\label{Proof_proposition_socialplanner}
\subsubsection{Formation of Social Planner's Problem }
\label{Section_formation_social_planner}
For the social planner, in period $\Delta$, if the current agent reports $\tilde c$ (we use $c$ to replace $\tilde c$ since in equilibrium we have truthfully reporting), his immediate gain from total social surplus, should be
\[\begin{aligned}
        \tilde \pi(\triangle,\alpha ,c)&=p(\alpha, c) \int_0^{\Delta} \lambda e^{-\lambda t} e^{-rt} \alpha z dt -p(\alpha, c)q(\alpha, c) \int_0^{\Delta} e^{-rt}c dt +\int_0^{\Delta}e^{-rt}[1-p(\alpha, c)]sdt \\&= \frac{1-e^{-(r+\lambda)\Delta}}{r+\lambda} \lambda p(\alpha, c)\alpha z -\frac{1-e^{-r\D}}{r} p(\alpha, c)q(\alpha, c)c+\frac{1-e^{-r\Delta}}{r}[1-p(\alpha, c)]s
\end{aligned}\]
from which we can get
\[\lim\limits_{\Delta\to 0}\frac{\tilde \pi(\Delta,\alpha ,c)}{\Delta}=\lambda p(\alpha,c)\alpha z -p(\alpha, c)q(\alpha, c)c+[1-p(\alpha, c)]s\]
We denote $\hat W(\alpha,c)$ as the value function of the social planner, contingent on knowing current agent's reporting cost $c$, and $W(\alpha) = \mathbb E_{c} \hat W(\alpha,c)$ as the expected value function. We denote $p=p(\alpha,\tilde c)$ and $q=q(\alpha,\tilde c)$ for simplicity, and again denote $\gamma (\alpha,\Delta)=1-\alpha(1-e^{-\lambda \Delta })$ for probability of reporting a zero utility. Then they must fit:
\[\hat W(\alpha,c)=\max_{p,q,t}\tilde \pi(\Delta,c)+pq e^{-r\Delta} \gamma(\alpha,\Delta)W(\alpha')+pq e^{-r\Delta} [1-\gamma(\alpha,\Delta)]\tilde V+(1-pq)e^{-r\Delta} W(\alpha)\]
and
\[W(\alpha)=\mathbb E_c  \hat W(\alpha, c)\]

Denote $A=\mathbb E_{c} (pq) (1-\gamma(\alpha,\Delta))$, then 
\[\begin{split}
0=\max_{p,q,t} \mathbb E_c\lim\limits_{\Delta \to 0}\frac{\tilde \pi(\Delta,\alpha ,c)}{\Delta}+\mathbb E_{c} (pq) \gamma W'(\alpha) (\alpha^2-\alpha)\lambda+\mathbb E_{c} (pq)\frac{\alpha z \lambda^2}{r}-\lim\limits_{\Delta \to 0} W(\alpha) \frac{1-e^{-r\Delta}+Ae^{-r\Delta}}{\Delta}
\end{split}\]
as a result, we get the following Bellman Equation:
\[rW(\alpha)=\lambda \alpha z \mathbb E_c(p)+[1-\mathbb E_c(p)] s-\mathbb E_c(pqc)+\frac{\lambda^2 z\alpha }{r} \mathbb E_c(pq)+(\alpha^2-\alpha)\lambda W'(\alpha)\mathbb E(pq)-\lambda \alpha \mathbb E_c(pq)W(\alpha)\]

\subsubsection{Social Planner's Optimization}
\label{proof_social_planner_optimization}

We should expect a pair of corner solution given the linear formation. To characterize the optimization solution, we first pin down the optimal $q(\alpha,\cdot)$. Put all terms to the right hand side, the terms contain $q$ are (and we should maximize it)
\[\mathbb E(qp) B(\alpha)-\lambda \mathbb E(qpc)\]
where $B(\alpha)\equiv \frac{\alpha\lambda^2 z}{r}+(\alpha^2-\alpha)\lambda W'(\alpha)-\lambda \alpha W(\alpha)$ is irrelevant to $c$.
\[\mathbb E(qp) B(\alpha)-\lambda \mathbb E(qpc)=\int_c q(\alpha,c)[B(\alpha) p(\alpha,c)-\lambda  p(\alpha,c) c] dH(c)\]
Hence to maximize it we have to set
\[q(\alpha,c)=\left \{
\begin{aligned}
	&1 &\text{ where } B(\alpha) -\lambda   c\ge 0\\
	&0 &\text{ otherwise } 
\end{aligned} \right.\]
Hence in optimal solutions, we have 
\[\mathbb E(qp) B(\alpha)-\lambda \mathbb E(qpc)=\int_c p(\alpha,c)\max\{B(\alpha) -\lambda  c,0\} dH(c)\]
and the original equations requires optimal $p(\alpha,\cdot)$ to maximize
\[rW(\alpha)-s=\max_p (\lambda z \alpha -s)\mathbb E(p)+\int_c p(\alpha,c)\max\{B(\alpha) -\lambda  c,0\} dH(c)\]

Hence if $\alpha>\frac{s}{\lambda z}$, optimal $p(\alpha,c)$ should be 1 for all $c$. Otherwise, it should be 1 where $\lambda z \alpha -s+B(\alpha) -\lambda c>0$. To conclude, the optimal allocation rule $p(\alpha,c)$ should take the form

\[p(\alpha,c)=\left \{
\begin{aligned}
	&1 &\text{ where } \max\{\lambda z \alpha -s,B(\alpha)-\lambda c+\lambda z \alpha -s\}>0\\
	&0 &\text{ otherwise } 
\end{aligned} \right.\]
Hence for $1>\alpha>\frac{s}{\lambda z}$,we have 
\[rW(\alpha)=\lambda z\alpha  +\int_{0<c<\frac{B(\alpha)}{\lambda}} B(\alpha)-\lambda c dH(c)\]
and for $0<\alpha<\frac{s}{\lambda z}$,we have 
\[rW(\alpha)-s=\int_{0<c<\frac{B(\alpha)+\lambda z\alpha -s}{\lambda}} \lambda z\alpha -s+B(\alpha)-\lambda c dH(c)\] 

Then we consider the optimal switching problem, the social planner will set $p=q=0$ to sell safe arm at some cutoff belief $\underline \alpha_5<\frac{s}{\lambda z}$, with $W(\underline \alpha_5)=\frac{s}{r}$ and $W'(\underline \alpha_5)=0$. Hence according to our definition, we have $B(\underline \alpha_5)=\frac{\underline\alpha_5\lambda (\lambda z-s)}{r}$. As a conclusion, at $\underline \alpha_5$ we have 
\[0=\int_{0<c<\frac{B(\underline \alpha_5)+\lambda z\alpha -s}{\lambda}} \lambda z\alpha -s+B(\underline \alpha_5)-\lambda c dH(c)\]
and a quick solution should be letting $B(\underline \alpha_5)+\lambda z\alpha -s=0$ hence we get 
\[\underline \alpha_5=\frac{rs}{rg+\lambda (g-s)} (=\underline \alpha_2)\]
which also fits the smooth pasting conditions.
This result means the social planner will stop at the exact same belief as the partial coverage monopolist.

When $\alpha>\underline \alpha_5$ the social planner's behavior is rather simple. According to definition before, if $\frac{s}{\lambda z}>\alpha>\underline \alpha_5$, we have $p(\alpha,c)=1$ if and only if $c<C_1(\alpha)\equiv \frac{B(\alpha)+\lambda z\alpha -s}{\lambda}$. Since at this time $\alpha<\frac{s}{\lambda z}$, we will always have $B(\alpha)-\lambda c>0$ (i.e. $c<C_2(\alpha)$) if $B(\alpha)+\lambda z\alpha -s>\lambda c$. Hence at any given $(\alpha,c)$, if $p(\alpha,c)=1$ we will always have $q(\alpha,c)=1$. 
This strategy is equivalent to monopolist's Partial Coverage strategy. And for $1>\alpha>\frac{s}{\lambda s}$, we have  $p(\alpha,c)=1$ as discussed before, and $q(\alpha,c)=1, \forall c<C_2(\alpha)=\frac{B(\alpha)}{r}$. In this case it's equivalent to monopolist's Full Coverage strategy. 

Of course there is also the contingency that Full Coverage will degenerate into Non Bonus case. In this language using Non Bonus is implemented by setting $p(\alpha,c)\equiv 1$ and $q(\alpha,c)\equiv 0$. Hence the FC wil again naturally degenerates into NB if $B(\alpha)\le 0$. At this switching point $W(\alpha)$ again is tangent to line $\alpha \lambda z$. It's straightforward that condition $ W(\bar \alpha_{FC}^{NB})=\frac{\bar \alpha_{FC}^{NB}\lambda z}{r}$ and $W'(\bar \alpha_{FC}^{NB})=\lambda z$ guarantees that 
$B(\bar \alpha_{FC}^{NB})=0$. 
In the formal context we use $\bar \alpha_{PC}^{FC}$, which is $\frac{s}{\lambda z}$ to stands for this switching cutoff between PC and FC for social planner . And we use $\bar \alpha_{SA}^{PC}=\underline \alpha_5$ to stands for this switching cutoff between SA and PC for social planner. 

An important part is to consider the monotonicity of $C_1,C_2$ with respect for $\alpha$.  First we should know that according to the definition, we have
\[\lambda C_1=\lambda z \alpha -s+B(\alpha)\] for all $\alpha\in (\alpha_{SA}^{PC},\alpha_{PC}^{FC})$.
And the Bellman Equation for $W(\alpha)$ in this interval should be like 

\[rW(\alpha)-s=\int_{0<c<C_1} \lambda C_1-\lambda c dH(c)  \] 

Take derivation with respect to $\alpha$ in both sides, in left hand side it's $\frac{d W(\alpha)}{d \alpha}>0$ and for the right hand side we have
\[\frac{\partial RHS}{\partial \alpha}=\frac{\partial RHS}{\partial C_1} \frac{\partial C_1}{\partial \alpha}\]
and 
\[\frac{\partial RHS}{\partial C_1}=\lambda H(C_1)>0\]
hence we have \[\frac{\partial C_1}{\partial \alpha}>0\]
As for $C_2$, we have $\lambda C_2=B(\alpha)$ and for $\alpha \in (\alpha_{PC}^{FC},\alpha_{FC}^{NB})$ we have 
\[rW(\alpha)-\lambda z\alpha =\int_{0<c<C_2} \lambda C_2-\lambda c dH(c)\]
we can prove $\frac{d rW(\alpha)-\lambda z\alpha}{d \alpha}<0$ as in section \ref{proof_fconly}, and prove $\frac{\partial RHS}{\partial C_2}>0$ as before. Hence we have $\frac{\partial C_2}{\partial \alpha}<0$. 
\subsubsection{Dynamic Mechanism Design}
The left part of the mechanism design problem is to design $t(\alpha,c)$ carefully to guarantee that IC holds.

\[U(c,\tilde c,\alpha)=-t(\alpha,\tilde c)+p(\alpha,\tilde c)[\alpha z \lambda -\lambda q(\alpha,\tilde c) c]+(1-p(\alpha,\tilde c))s\]

The following lemma shows that above social mechanism design is feasible using transfers:
\begin{lemma}
	The following table \ref{social mechanism} holds.
\end{lemma}
\begin{table}[H]
\centering
	\begin{tabular}{|c|c|c|c|c|}
		\hline
		condition1	& condition2 &  $p(\alpha,c)$ & $q(\alpha,c)$ & $t(\alpha,c)$ \\
		\hline
		$\alpha>\frac{s}{\lambda z}$	&  $0<c<C_2(\alpha)$ & 1 & 1 & $\alpha z \lambda -\lambda C_2(\alpha)$ \\
		\hline
		$\alpha>\frac{s}{\lambda z}$	&  $C_2(\alpha)<c<C_1(\alpha)$& 1 & 0 & $\alpha z\lambda $ \\
		\hline
		$\alpha>\frac{s}{\lambda z}$	&  $C_1(\alpha)<c<\bar c$& 0 & - & $s$ \\
		\hline
		$\alpha<\frac{s}{\lambda z}$	&  $0<c<C_1(\alpha)$& 1 & 1 & $\alpha z \lambda -\lambda C_1(\alpha)$ \\
		\hline
		$\alpha<\frac{s}{\lambda z}$	&  $C_1(\alpha)<c<\bar c$& 0 & - & $s$ \\
		\hline
	\end{tabular}
	\caption{Optimal mechanism of social planner}
	\label{social mechanism}
\end{table}

\begin{proof}

Since $\alpha$ is common knowledge, we only need to guarantee that the agent will report $c$ truthfully. In this case, $C_1(\alpha),C_2(\alpha)$ and $B(\alpha)$ is a known constant.

In our ideal setting of $p(\alpha,c)$ and $q(\alpha,c)$, there are three possible combinations:
\[\left \{
\begin{aligned}
    &p=0 & \Rightarrow &u_1=-t(\alpha,c)+s\\
    &p=1,q=0  &\Rightarrow &u_1=-t(\alpha,c)+\alpha \lambda z\\
    &p=1,q=1 &\Rightarrow &u_3=-t(\alpha,c)+\lambda z \alpha-c \\
\end{aligned} \right.\]
We now prove this can be supported by constant transfer. If $\alpha<\frac{s}{\lambda z}$, we let $t(\alpha,c)=s,\forall c>C_1$. So by assuming $\forall c>C_1$, $ t(\alpha,c)$ is some constant $M$, we require
\[0\ge -M+\alpha\lambda z-c,\forall c >C_1\]
and 
\[-M+\alpha\lambda z-c \ge 0,\forall c \le C_1\]
    Hence $M=\alpha \lambda z-C_1$ satisfies our requirements.  

\[t(\alpha,c)= \left \{
\begin{aligned}
    & s & \text{ if } \alpha<\frac{s}{\lambda z}, c>C_1(\alpha)\\
    & \alpha \lambda z-\lambda C_1(\alpha) & \text{ if } \alpha<\frac{s}{\lambda z}, c<C_1(\alpha)\\
\end{aligned} \right.\]

Similarly if $\alpha >\frac{s}{\lambda z}$ we can always let 
\[t(\alpha,c)= \left \{
\begin{aligned}
    & s & \text{ if } \alpha>\frac{s}{\lambda z}, c>C_1(\alpha)\\
    & \alpha \lambda z & \text{ if } \alpha>\frac{s}{\lambda z}, C_2(\alpha)<c<C_1(\alpha)\\
    & \alpha \lambda z-\lambda C_2(\alpha) & \text{ if } \alpha>\frac{s}{\lambda z}, c<C_2(\alpha)\\
\end{aligned} \right.\]
Since for $ c>C_1(\alpha)$, truthful reporting generates utility 0. But deviating to any $\tilde c<C_1(\alpha)$ also generates utility 0, while deviating to $\tilde c<C_2(\alpha)$ generates utility $\lambda (C_2-c)<0$. $ C_1(\alpha)>c>C_2(\alpha)$ works with the same logic. And for $c<C_2(\alpha)$, truthful reporting (and deviating to any $\tilde c<C_2(\alpha)$) generates utility $\lambda (C_2-c)>0$. Otherwise the utility is 0.
\end{proof}
\subsubsection{Social Surplus under Monopolist}
\label{Social_Surplus_Monopolist}
We calculate the value function of social surplus under monopolist setting. The social surplus has the following stage payoff:
\[
    \lim_{\Delta\to0}\frac{\pi_{L}(\Delta,c)}{\Delta}=\left\{
\begin{aligned}
	&\alpha z \lambda - P_L(b)c& \text{ if } L=\text{FC}\\
	&P_L(b) \lambda z \alpha-P_L(b)c +[1-P_L(b)]s& \text{ if } L=\text{PC}\\
	& s&\text{ if } L=\text{SA}\\
    & \alpha z  \lambda &\text{ if } L=\text{NB}
\end{aligned} \right.\]

The law of motion like \eqref{ODE_vf} still holds, except for the fact that the social surplus is not a result of maximization. Accordingly, the bonus is not optimally determined, instead it should be determined as in Lemma \ref{fconly_lemma} and \ref{pconly_lemma}. To be more specific, denote the value function when the belief is $\alpha$ and the current agent's reporting cost is $c$ as $\Gamma(\alpha,c)$. Note that $\mathbb E_c[P_L(b)]=H(b)$ and $\mathbb E_c[P_L(b)c]=H(b) \mathbb E(c|c<b)$. Let $\Lambda(\alpha)=\mathbb E_{c}[\Gamma(\alpha,c)]$ to stand for the value function of social surplus under the monopolist setting; then we have
\[
	[r+P_L(b)\alpha \lambda ] \Gamma(\alpha,c)=\lim_{\Delta\to0}\frac{\pi_{L}(\Delta,c)}{\Delta}+P_L(b)(\alpha^2-\alpha)\lambda \Gamma_{\alpha}(\alpha,c)+P_L(b)\frac{\alpha \lambda^2 z }{r}
\]
hence 
\begin{equation}
    [r+H(b)\alpha \lambda ] \Lambda(\alpha)=\mathbb E_c\lim_{\Delta\to0}\frac{\pi_{L}(\Delta,c)}{\Delta}+H(b)(\alpha^2-\alpha)\lambda \Lambda'(\alpha)+H(b)\frac{\alpha \lambda^2 z }{r}
    \label{value_function_social_surplus}
\end{equation}
Notice in \eqref{value_function_social_surplus}, the bonus $b$ is not results of optimization. Instead, $b(\alpha)$ is determined by FOC of optimal switching monopolist. If $\alpha\in [\alpha_{PC}^{FC},\alpha_{FC}^{NB}]$, it fits
\[\beta(b)=\frac{\alpha z \lambda^2 }{r}+(\alpha^2-\alpha)\lambda V'(\alpha)-\alpha\lambda V(\alpha)\]
and If $\alpha\in [\alpha_{SA}^{PC},\alpha_{PC}^{FC}]$, it fits
\[ \beta(b)=\lambda \alpha z -  s+\frac{\alpha z \lambda^2}{r}+(\alpha^2-\alpha)\lambda U'(\alpha)-\lambda \alpha U(\alpha)\]
where $V(\cdot)$ and $U(\cdot)$ are value functions of optimally switching monopolist, and 
\[\mathbb E_c\lim_{\Delta\to0}\frac{\pi_{1_,L}(\Delta,c)}{\Delta}=\left\{
\begin{aligned}
	&\alpha z \lambda - H(b)\mathbb E(c|c<b)& \text{ if } L=\text{FC}\\
	&H(b) \lambda z \alpha-H(b)\mathbb E(c|c<b) +[1-H(b)]s& \text{ if } L=\text{PC}\\
	& s&\text{ if } L=\text{SA}\\
    & \alpha z  \lambda &\text{ if } L=\text{NB}
\end{aligned} \right.\]

\subsection{Proof of Proposition \ref{proposition_speed}}
\begin{proof}
The second half is actually proven in Section \ref{proof_social_planner_optimization}. 
(Not updated, need to add)
In order to show our results more intuitively, we characterize the equivalent mechanism $(p,q)$ of the monopolist.
Take the partial coverage case as an example, the following setting can lead the (ex ante) equivalent learning process as monopolist's partial coverage 
\[p(\alpha,c)=\left \{
\begin{aligned}
	&1 &\text{ where } R(\alpha)-c>0\\
	&0 &\text{ otherwise } 
\end{aligned} \right.\]
and 
\[q(\alpha,c)\equiv 1\]
where as declared in Lemma \ref{pconly_lemma}, $R(\alpha)$ is the solution of equation 
\[
r\beta(R)=\lambda \alpha z r- r s+\alpha z \lambda^2+(\alpha^2-\alpha)\lambda U'(\alpha)-\lambda \alpha U(\alpha)
\]
Hence at the stopping belief $\underline \alpha_2$, we should have $p(\underline \alpha_2,c)=0, \forall c$, hence $R(\alpha)=0$ which leads to the restriction of value function:
\[
0=\lambda \underline \alpha_2 z r- r s+\underline \alpha_2 z \lambda^2+(\underline \alpha_2^2-\underline \alpha_2)\lambda U'(\underline \alpha_2)-\lambda \alpha U(\underline \alpha_2)
\]
which is just the same as the value function of social planner when $\alpha<\frac{s}{z\lambda }$. Hence in the monopolist's model, in the interval of using Partial Coverage, he sets $p=1$ for all $c$ such that 
\[
rc<\beta^{-1}[\lambda \alpha z r- r s+\alpha z \lambda^2+(\alpha^2-\alpha)\lambda U'(\alpha)-\lambda \alpha U(\alpha)]
\]
But in social planner's model, he sets $p=1$ for all $c$ such that
\[r c<\lambda \alpha z r- r s+\alpha z \lambda^2+(\alpha^2-\alpha)\lambda W'(\alpha)-\lambda \alpha W(\alpha)\]
Since $\beta(x)=\frac{H(x)}{h(x)}+x>x$, we always have the reporting cost cutoff $\bar C_{\text{mo}}< \bar C_{\text{sp}}$. Hence we say the monopolist will always set higher probability to do experiments.

Similarly, using Full Coverage is equivalent to set
\[p(\alpha,c)\equiv 1\] and 
\[q(\alpha,c)=\left \{
\begin{aligned}
	&1 &\text{ where } b(\alpha)-c>0\\
	&0 &\text{ otherwise } 
\end{aligned} \right.\]
According to our definition, $b(\alpha )$ is the solution of equation
\[
r\beta(b)=\alpha z \lambda^2 +(\alpha^2-\alpha)\lambda V'(\alpha)-\alpha\lambda V(\alpha)
\]

Simply comparing $(p,q)$ combination between the mechanism under the monopolist and the social planner we can find the social planner will do experiments in a higher probability hence attaches the termination in a higher speed. 
\end{proof}

\subsection{Numerical Example in Section \ref{section_Imperfect_Learning}}
\label{appendix_numeric}
 In this section we provide a numerical example by assuming $c_i\sim U[0,1]$. 
We set $r=0.5,z=7,\lambda=0.8, s=0.5\lambda z$. Hence according to Lemma  \ref{fconly_lemma}, the ODE of FC only value function should be
 \[rV(\alpha)=\lambda z \alpha+ b^2\]
 and
 \[2 b=(\alpha^2-\alpha)\lambda V'(\alpha)-\alpha\lambda V(\alpha)+\frac{\alpha \lambda^2 z}{r}\]

 Hence we have 
 \[4b^2=4[rV(\alpha)-\lambda z \alpha]\]
 So 
 \[2b=\sqrt{4 [rV(\alpha)-\lambda z \alpha]}\]
 So we have the law of motion 
 \[V'(\alpha) = \frac{\sqrt{4 [rV(\alpha)-\lambda z \alpha]}+\alpha \lambda V(\alpha)-\frac{\alpha \lambda^2 z}{r}}{(\alpha^2-\alpha)\lambda}\]

For boundary condition we need 
$\sqrt{4[s-\lambda z \alpha]}+\alpha \lambda \frac{s}{r}-\frac{\alpha \lambda^2 z}{r} = 0$. 
Hence at $\underline \alpha_1= \frac{-2\lambda z+2\sqrt{\lambda^2 z^2+A^2s}}{A^2}$ where $A=\frac{\lambda(g-s)}{r}$ we have $V(\alpha)=s/r$. 

Similarly, with Lemma \ref{pconly_lemma} we know $U(\alpha)$ fits the following ODE
 \[U'(\alpha) = \frac{\sqrt{4 [rU(\alpha)-s]}+\alpha \lambda U(\alpha)-\frac{\alpha \lambda^2 z}{r}+s-\lambda \alpha z}{(\alpha^2-\alpha)\lambda}\]
 with boundary condition
 \[U(\frac{rs}{\lambda(g-s)+rg})=\frac{s}{r}\]

 In optimal switching case, we first pin down $U(\alpha)$ and hence get $U(0.5)$. Then for $V(\alpha)$ still fits the original law of motion
\[V'(\alpha) = \frac{\sqrt{4 [rV(\alpha)-\lambda z \alpha]}+\alpha \lambda V(\alpha)-\frac{\alpha \lambda^2 z}{r}}{(\alpha^2-\alpha)\lambda}\]
  with boundary condition 
  \[V(0.5)=U(0.5)\] 
  Hence we can pin down $V(\alpha)$. Denote the intersection point of $V(\alpha)$ and $\frac{\alpha z \lambda}{r}$ we get  $\alpha_{FC}^{NB}$.

  In social planner case, if $\alpha<\alpha_{PC}^{FC}$, we have
  \[rW(\alpha)-s=\int_{0<c<C_1} \lambda C_1-\lambda cd H(c)=\frac{\lambda }{2}C_1^2\]
  hence 
  \[\lambda C_1=\sqrt{2\lambda(rW(\alpha)-s)} \]
  plug the definition of $C_1$, we can get
  the motion law is 
  \[W'(\alpha)=\frac{s-\alpha z \lambda+\lambda \alpha W(\alpha)-\frac{\alpha \lambda^2 z}{r}+\sqrt{2\lambda(rW(\alpha)-s)} }{(\alpha^2-\alpha)\lambda }\]
  with boundary condition 
  \[W(\alpha_{SA}^{PC})=\frac{s}{r}\]
  Similarly when $\alpha>\alpha_{PC}^{FC}$, we have
  \[W'(\alpha)=\frac{\lambda \alpha W(\alpha)-\frac{\alpha \lambda^2 z}{r}+\sqrt{2\lambda(rW(\alpha)-\lambda \alpha z)} }{(\alpha^2-\alpha)\lambda }\]
  with boundary condition requiring $W(\alpha)$ is smooth at $\alpha_{PC}^{FC}$. 

  Then we calculate the social surplus under monopolist's setting. If $c\sim U[0,1]$ then $E(c|c<x)H(x)=\frac{x^2}{2}$. Hence in PC interval, we have 
  \[r\Lambda(\alpha)-s=\frac{3}{2}\lambda R^2\]
  hence 
  \[2\lambda R=2\sqrt{\frac{2}{3}(r\Lambda(\alpha)-s)}\]
  hence 
  \[\Lambda'(\alpha)=\frac{2\sqrt{\frac{2}{3}(r\Lambda(\alpha)-s)}+\lambda \alpha \Lambda(\alpha)+(s-\lambda z\alpha)-\frac{\alpha z \lambda^2}{r}}{(\alpha^2-\alpha)\lambda}\]
  with the same boundary condition of $U(\alpha)$. Similarly in Full Coverage Interval we have
  \[\Lambda'(\alpha)=\frac{2\sqrt{\frac{2}{3}(r\Lambda(\alpha)-\alpha\lambda z)}+\lambda \alpha \Lambda(\alpha)-\frac{\alpha z \lambda^2}{r}}{(\alpha^2-\alpha)\lambda}\]
 with the boundary condition $\Lambda(\alpha)$ is smooth at $\alpha_{FC}^{PC}$. \

 \subsection{Numerical Example in Section \ref{section_IR}}
 \label{section_numerical_IR}
 For Figure\ref{fig:IRpqplot}, the parameter setting is $r=0.5, z=20, \lambda=0.8, s=0.5g,
\bar c=1$. For Figure \ref{fig:IRvalue_function}, the parameter setting is $r=0.5, z=20, \lambda=0.8, s=0.5g,
\bar c=0.8$.

The algorithm is, first get $U(\alpha)$ and $V(\alpha)$ as if $\bar c$ is large enough (i.e., use $H(b)=\frac{b}{\bar c}$ instead of  $H(b)=\min(\frac{b}{\bar c},1)$. The basic algorithm are the same as in Section \ref{appendix_numeric}, while the only difference is to modify the law of motion allowing arbitrary $\bar c$ as:
 \[V'(\alpha) = \frac{\sqrt{4 \bar c[rV(\alpha)-\lambda z \alpha]}+\alpha \lambda V(\alpha)-\frac{\alpha \lambda^2 z}{r}}{(\alpha^2-\alpha)\lambda}\]
and 
  \[U'(\alpha) = \frac{\sqrt{4\bar c [rU(\alpha)-s]}+\alpha \lambda U(\alpha)-\frac{\alpha \lambda^2 z}{r}+s-\lambda \alpha z}{(\alpha^2-\alpha)\lambda}\]

Then based on Lemma \ref{fconly_lemma} and \ref{pconly_lemma}, we calculate $b_V(\alpha)$ and $b_U(\alpha)$ separately. Since there is no explicit expression of solution value function, we use polynomial approximation. Then define the real bonus like $\tilde b_K(\alpha)=\min(b_K(\alpha),\bar c)$ for $K\in \{U,V\}$, and get $\alpha_{PC}^{IR}$, $\alpha_{IR}^{FC}$ simultaneously. Then given $\tilde b_U(\alpha)$, we can get real $U(\alpha)$ by ODE:
\begin{equation*}
    \begin{aligned}
    r U(\alpha)=&H(\tilde b_U(\alpha)) (\lambda \alpha z-\tilde b_U(\alpha)) +[1-H(\tilde b_U(\alpha))]s\\&+H(\tilde b_U(\alpha))\frac{\alpha \lambda^2 z}{r}+H(\tilde b_U(\alpha))(\alpha^2-\alpha)\lambda U'(\alpha)-H(\tilde b_U(\alpha))\alpha \lambda U(\alpha) 
\end{aligned}
\end{equation*}
with boundary condition
\[U(\alpha_{SA}^{PC})=\frac{s}{r}, \alpha_{SA}^{PC}=\frac{sr}{\lambda(g-s-zr)}\]
Then for $Y(\cdot)$ we have
\[    r Y(\alpha)= \lambda \alpha z(1+ \frac{\lambda}{r})-\bar c-\alpha \lambda Y(\alpha)+(\alpha^2-\alpha)\lambda Y'(\alpha)\]
with 
\[Y(\alpha_{PC}^{IR})=U(\alpha_{PC}^{IR})\]
Then for $V(\cdot)$ we have
 \begin{equation*}
 \begin{aligned}
     r V(\alpha)=&\lambda \alpha z[1+H(\tilde b_V(\alpha) )\frac{\lambda}{r}]-H(\tilde b_V(\alpha) )\tilde b_V(\alpha) \\&-H(\tilde b_V(\alpha) )\alpha \lambda V(\alpha)+H(\tilde b_V(\alpha) )(\alpha^2-\alpha)\lambda V'(\alpha)
 \end{aligned}
\end{equation*}
with 
\[V(\alpha_{IR}^{FC})=Y(\alpha_{IR}^{FC})\]
Finally $V(\alpha)$ smoothly switch to Non Bonus when $V'(\alpha_{FC}^{NB})=\frac{\lambda z}{r}$. 
 \newpage

\end{document}